\documentclass[journal]{IEEEtran}

\usepackage{amsmath,amssymb,amsthm, mathtools, graphicx}
\usepackage{listings,color,caption}
\usepackage{amsthm}
\usepackage{amsmath}
\usepackage{bm}
\usepackage{cite}

\usepackage{verbatim}

\usepackage{algpseudocode}
\usepackage{algorithm}

\newtheorem{assumption}{\bf Assumption}
\newtheorem{definition}{\bf Definition}
\newtheorem{theorem}{\bf Theorem}
\newtheorem{proposition}{\bf Proposition}
\newtheorem{lemma}{\bf Lemma}

\newtheorem{remark}{\bf Remark}

\DeclareMathOperator*{\argmin}{arg\,min}

\begin{document}

\title{Sequential learning and control:\\Targeted exploration for robust performance}

\author{Janani Venkatasubramanian$^1$, Johannes K\"ohler$^2$, Julian Berberich$^1$, Frank Allg\"ower$^1$
\thanks{$^1$ Janani Venkatasubramanian, Julian Berberich and Frank Allg\"ower are with the Institute for Systems Theory and Automatic Control, University of Stuttgart, 70550 Stuttgart, Germany. (email:\{janani.venkatasubramanian, julian.berberich, frank.allgower\}@ist.uni-stuttgart.de)

$^2$ Johannes K\"ohler is with the Institute for Dynamic Systems and Control, ETH Z\"urich, Z\"urich CH-80092, Switzerland. (email:jkoehle@ethz.ch)}
\thanks{ F. Allg\"ower is thankful that his work was funded by Deutsche Forschungsgemeinschaft (DFG, German Research Foundation) under Germany’s Excellence Strategy - EXC 2075 - 390740016 and under grant 468094890. F. Allg\"ower acknowledges the support by the Stuttgart Center for Simulation Science (SimTech). Janani Venkatasubramanian thanks the International Max Planck Research School for Intelligent Systems (IMPRS-IS) for supporting her.}}

\maketitle

\begin{abstract}
We present a novel dual control strategy for uncertain linear systems based on targeted harmonic exploration and gain-scheduling with performance and excitation guarantees. In the proposed sequential approach, robust control is implemented after exploration with the main feature that the exploration is optimized with respect to the robust control performance. Specifically, we leverage recent results on finite excitation using spectral lines to determine a high probability lower bound on the resultant finite excitation of the exploration data. This provides an a priori upper bound on the remaining model uncertainty after exploration, which can further be leveraged in a gain-scheduling controller design that guarantees robust performance. This leads to a semidefinite program-based design which computes an exploration strategy with finite excitation bounds and minimal energy, and a gain-scheduled controller with probabilistic performance bounds that can be implemented after exploration. The effectiveness of our approach and its benefits over common random exploration strategies are demonstrated with an example of a system which is ‘hard to learn'.

\end{abstract}


\section{Introduction}

Control inputs to an uncertain system should have a `directing effect' to control the dynamical system and achieve a certain goal. Furthermore, the inputs should also have a `probing' effect to learn the uncertainty in the system. These two effects are, however, competing and draw attention to the exploration - exploitation trade-off, which is also the subject of contemporary literature on Reinforcement Learning \cite{recht2019tour}. In close association with these effects, simultaneous learning and control of uncertain dynamic systems has garnered research interest with the establishment of the \emph{dual control} paradigm \cite{feldbaum1960dual}. A detailed survey of dual control methods is provided in \cite{filatov2000survey, mesbah2018stochastic}. These methods laid the foundation of balancing exploration with caution.

Early works in this direction involve approximations of stochastic Dynamic Programming (DP), known as implicit dual control methods, to handle  computational intractability problems that accompany stochastic DP \cite{tse1973wide, tse1976actively, bayard1985implicit, bar1976caution}. Another class of dual control, known as explicit dual control methods, use heuristic probing techniques to actively learn the uncertain system \cite{wittenmark1995adaptive}, and are closely related to \emph{Optimal Experiment Design} \cite{geversaa2005identification, hjalmarsson2005experiment}. These methods solve the combined problem of regulating and probing closed-loop system dynamics, and have been further refined to explicitly consider the relevance of the uncertainty in the cost function, yielding application-oriented strategies \cite{annergren2017application, larsson2016application, heirung2017dual}. 
 
A particularly appealing approach to the dual control problem is to sequentially apply some probing input for exploration, and then design a robustly stabilizing feedback based on the gathered data. Recent methods focus on \emph{targeted} exploration, and perform better than methods that use common greedy random exploration \cite{barenthin2008identification, bombois2006least, bombois2021robust, umenberger2019robust, ferizbegovic2019learning, iannelli2019structured}. In particular, these methods consider that exploration should be targeted in the sense that the resulting uncertainty reduction in the model facilitates achieving a control goal and performance objective. Identification strategies proposed in \cite{bombois2006least, bombois2021robust, barenthin2008identification} yield uncertainty sets within required bounds with possibly small experiment cost. In these works, the experiment cost is determined in terms of experiment time, performance degradation due to experimentation, and/or power of the input. In \cite{barenthin2008identification}, identification and robust control are jointly designed. Iterations of identification and robust control design are employed to ensure the uncertain parameters meet a desired accuracy that is required for robust control. As a result, the desired properties for robust control can only be ensured after iterative experiments in \cite{barenthin2008identification}.

Developing on the works in \cite{barenthin2008identification} and \cite{dean2017sample}, a dual control strategy is proposed in \cite{ferizbegovic2019learning} that minimizes the worst-case cost achieved by a robust controller that is synthesized with reduced model uncertainty. However, the uncertainty bounds utilized in \cite{ferizbegovic2019learning} require data to be independent, and hence not applicable to time-series data available through exploration. In contrast, uncertainty bounds on the model parameters directly applicable to correlated time-series data are introduced in \cite{umenberger2019robust}. These bounds are used to design a targeted exploration strategy that excites the system to reduce uncertainty specifically to improve a robust LQR design in \cite{umenberger2019robust}. The approach in \cite{iannelli2019structured} extends the exploration strategy in \cite{umenberger2019robust} to a more realistic finite horizon problem setting that captures the trade-offs between exploration and exploitation better.


The exploration methods proposed in \cite{umenberger2019robust, ferizbegovic2019learning, iannelli2019structured} use a linear state-feedback and an additional Gaussian noise term for exploration. To tractably compute the predicted uncertainty bound associated with the parameter estimates after exploration, the empirical covariance is approximated by the worst-case state covariance. This approximation fails to provide a priori guarantees of excitation on the exploration inputs. While the results of methods in \cite{umenberger2019robust, ferizbegovic2019learning, barenthin2008identification, iannelli2019structured} seem to perform well numerically, they lack the corresponding theoretical performance guarantees. In particular, changes in the mean of uncertain system parameters during the exploration phase are not accounted for in the methods in \cite{umenberger2019robust, ferizbegovic2019learning}. We address these limitations with guaranteed finite excitation in combination with  our recently proposed gain-scheduled controller \cite{venkatasubramanian2020robust}, which accounts for changing parameter estimates and guarantees closed-loop performance.

In particular, we propose harmonic exploration inputs in the form of a linear combination of sinusoids of specific frequencies and optimized amplitudes, and reduce uncertainty in a targeted fashion with the goal of guaranteed control performance. This choice is also supported in literature, where it was established that the robust optimal control input can be expressed with appropriately chosen amplitudes and frequencies of the sinusoids \cite{rojas2007robust}. As one of our main contributions, we derive an \textit{a priori} guaranteed lower bound on the finite excitation of the exploration inputs depending on the spectral content of the inputs \cite{sarker2020parameter}. The lower bound on excitation results in a bound on the uncertain system parameters, which can be leveraged in robust control design to provide performance guarantees. Regarding the application of targeted exploration in dual control, we employ the exploration strategy in our recently proposed gain-scheduling-based dual control approach \cite{venkatasubramanian2020robust}. In particular, we treat the mean of future uncertain system parameters as a scheduling variable. The approach based on gain-scheduling gives rise to a design based on a semidefinite program (SDP) with closed-loop performance guarantees. The resulting controller is a state feedback which depends on the parameter estimates after the exploration phase and thus, on the data collected during this phase. The controller \emph{guarantees} robust closed-loop performance after an initial exploration phase.

%
 
\section{Problem Statement}
\subsubsection*{Notation} 
The transpose of a matrix $A \in \mathbb{R}^{n \times m}$ is denoted by $A^\top$. The conjugate transpose of a matrix $A \in \mathbb{C}^{n \times m}$ is denoted by $A^\mathsf{H}$. The positive definiteness of a matrix $A \in \mathbb{C}^{n \times m}$ is denoted by $A = A^\mathsf{H} \succ 0$. The Kronecker product operator is denoted by $\otimes$. The operator $\mathrm{vec}(A)$ stacks the columns of $A$ to form a vector. The operator $\mathrm{diag}(A_1,\dots,A_n)$ creates a block diagonal matrix by aligning the matrices $A_1,\dots,A_n$ along the diagonal starting with $A_1$ in the upper left corner. The critical value of the Chi-squared distribution with $n$ degrees of freedom and probability $p$ is denoted by $\chi_n^2(p)$. A unit sphere of dimension $d$ is denoted by $\mathcal{S}^{d-1}$. The expected value of a random variable $X$ is denoted by $\mathbb{E}(X)$. The probability of an event $E$ occurring is denoted by $\mathbb{P}(E)$. Given a sequence $\{ x_k\}_{k=0}^{T-1}$, the discrete Fourier transform (DFT) of the sequence is denoted by $\mathbf{x}(e^{j\omega})=\sum_{k=0}^{T-1}x_k e^{-j2\pi k \omega}$ where $\omega \in \Omega_T:=\{0,1/T,\dots,(T-1)/T\}$. For a vector $x\in\mathbb{R}^n$ and a matrix $P \succ 0$, the Euclidean norm is denoted by $\|x\|=\sqrt{x^\top x}$ and $\|x\|_P=\sqrt{x^\top P x}$. For a matrix $A\in\mathbb{C}^{m\times n}$, $\|A\|$ denotes the largest singular value. Furthermore, given a matrix $M \succeq 0$, $\|A\|_M = \| M^{1/2} A \|$. A random variable $X \in \mathbb{R}^d$ that is normally distributed with mean $\mu$ and variance $\Sigma$ is denoted by $X \sim \mathcal{N}(\mu,\Sigma)$. 
%


A random variable $X \in \mathbb{R}$ is said to be sub-Gaussian \cite{rigollet2015high}, with variance proxy $\sigma^2 \in \mathbb{R}$, i.e., $X \sim \mathsf{subG}(\sigma^2)$, if $\mathbb{E}(X)=0$ and its moment generating function satisfies
\begin{equation}
\mathbb{E}\left(e^{sX}\right) \leq e^{\left( \frac{\sigma^2 s^2}{2} \right)},\quad \forall s \in \mathbb{R}.
\end{equation}

A random vector $X \in \mathbb{R}^d$ is said to be sub-Gaussian with variance proxy $\sigma^2$, i.e., $X \sim \mathsf{subG}(\sigma^2)$, if $\mathbb{E}[X]=0$ and $u^\top X$ is sub-Gaussian with variance proxy $\sigma^2$ for any unit vector $u \in \mathcal{S}^{d-1}$.

\subsection{Setting}
Consider a discrete-time linear time-invariant dynamical system of the form \begin{equation}\label{sys}
x_{k+1}=A_{\mathrm{tr}}x_k+B_{\mathrm{tr}}u_k+w_k,\quad w_k \overset{\text{i.i.d.}}{\sim} \mathcal{N}(0,\sigma_\mathrm{w}^2I_{n_\mathrm{x}})
\end{equation}
where $x_k \in \mathbb{R}^{n_\mathrm{x}}$ is the state, $u_k \in \mathbb{R}^{n_\mathrm{u}}$ is the control input, and $w_k \in \mathbb{R}^{n_\mathrm{x}}$ is the normally distributed process noise. It is assumed that the realizations of $w_k$ are independent and identically distributed (i.i.d.) with zero mean and known variance $\sigma_\mathrm{w}^2I_{n_\mathrm{x}}$, and the state $x_k$ is directly measurable, allowing for the application of standard identification results\footnote{In case $w \sim \mathcal{N}(0,\Sigma)$ where $\Sigma$ is a known positive definite matrix, a `whitening' transformation may be applied such that in the transformed state space the covariance matrix of the disturbance is an identity matrix.} \cite{ljung1999sysid}. The true values of the system parameters $A_{\mathrm{tr}}$, $B_{\mathrm{tr}}$, are initially uncertain, and there is a need to gather informative data to improve the accuracy of the parameters.

\subsubsection*{Control goal}
The overarching goal of the proposed dual control strategy is to design a stabilizing state-feedback controller $u_k=K x_k$ which ensures that the closed-loop system is stable while also satisfying some desired performance specifications with high probability \cite{scherer2001lpv, veenman2014synthesis}. Consider the output $z_k \in \mathbb{R}^{n_\mathrm{z}}$ that depends on the state:
\begin{equation}\label{gz}
z_k=Cx_k,
\end{equation}
where $C$ is known. Note that $z_k$ is the closed-loop trajectory resulting from the application of the input $u_k=K x_k$.
 In the presence of white noise input signals, the $H_2$-norm provides a suitable stochastic interpretation in terms of the asymptotic output variance of the closed-loop system, i.e., $\lim_{k\rightarrow \infty} \mathbb{E} (\|z_k\|^2)$ \cite{paganini2000linear}.
\begin{definition}\label{def:h2perf}($H_2$ performance) The closed-loop system \eqref{sys} with $u_k=Kx_k$ and initial state $x_0=0$ achieves the $H_2$ performance bound $\gamma_\mathrm{p}$ if 
\begin{align}\label{eq:h2perf}
    \lim_{k\rightarrow \infty} \mathbb{E} (\|z_k\|^2) < \gamma_\mathrm{p}^2 \sigma_\mathrm{w}^2.
\end{align}
\end{definition}
Due to the uncertain dynamics, we cannot directly design a linear feedback $u_k=Kx_k$ that systematically accounts for the change in the mean estimate after exploration and achieves a desired $H_2$ performance bound $\gamma_\mathrm{p}$ (Def. \ref{def:h2perf}). Instead, we propose a sequential dual control approach wherein a targeted exploration strategy is implemented first, which is followed by the implementation of a parametrized state-feedback achieving \eqref{eq:h2perf}. To this end, we first provide preliminaries regarding uncertainty bounds for parameter estimation based on time-series data and spectral information in Section \ref{sec:prelim}. As one of the main results, we derive the exploration strategy, and the corresponding uncertainty bound on the data obtained during exploration in Section \ref{sec:exploration}. The primary challenge is to encapsulate the dual effect of the performance improvement through the process of exploration, and to tailor the exploration in a manner that is pertinent to performance improvement. Therefore, we simultaneously design the targeted exploration strategy and a parametrized state-feedback controller for the system in \eqref{sys} in dependence of the future parameters/model estimate. The new parameter estimate, which influences the state-feedback control law $K$, is interpreted as a scheduling variable using tools from gain-scheduling control \cite{venkatasubramanian2020robust}. The gain-scheduling controller is presented in  Section \ref{sec:gain_scheduling_h2}. Consequently, we solve a joint dual control problem of obtaining the exploration strategy and the state-feedback controller to be implemented after exploration without re-design, as elaborated in Section \ref{sec:dual_control}. A sketch of the overall dual control strategy is provided in Fig. \ref{fig:timeline}. The design process ensures that, by applying the feedback after the phase of exploration, $H_2$ performance specifications are guaranteed with high probability.

\begin{figure}[hbtp]
\begin{center}
\includegraphics[width=0.48\textwidth]{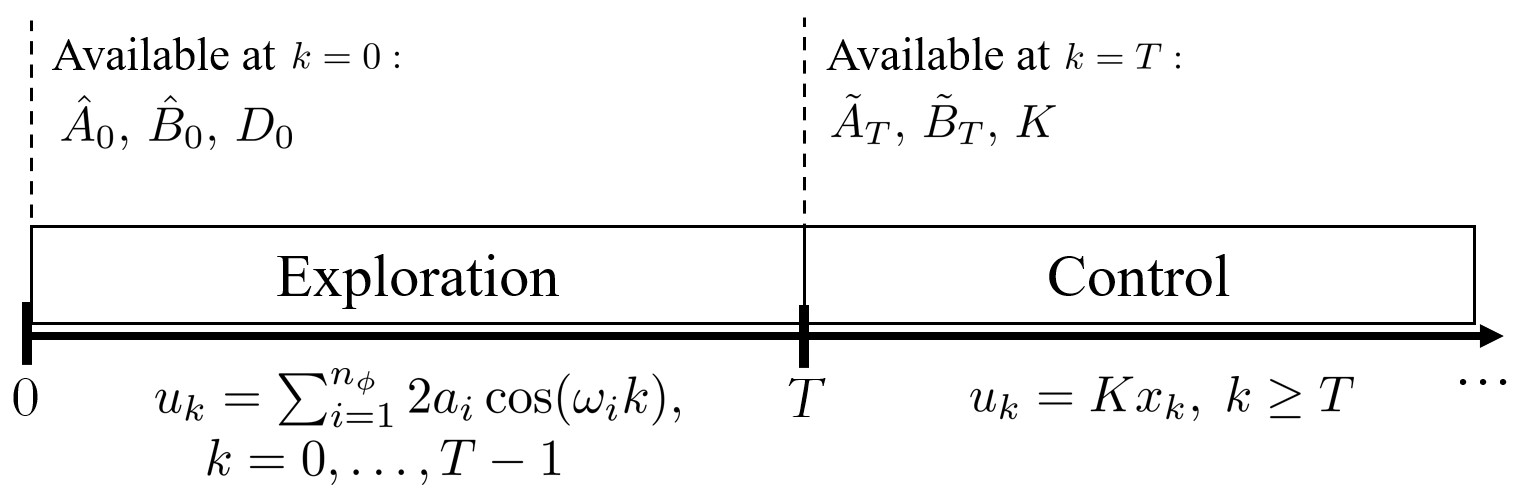}
\end{center}
\caption{Sketch of the sequential robust dual control strategy.}
\label{fig:timeline}
\end{figure}

\section{Preliminaries}\label{sec:prelim}
This section discusses prior results that provide a data-dependent uncertainty bound on the parameter estimates, and the required preliminaries regarding the theory of spectral lines, from \cite{umenberger2019robust} and \cite{sarker2020parameter}, respectively.

\subsection{Uncertainty bound}\label{sec:uncertainty_bound}
Given observed data $\mathcal{D}_T=\{x_k,u_k\}_{k=0}^{T-1}$ of length $T$, we are interested in quantifying the uncertainty associated with unknown parameters $A_{\mathrm{tr}}$ and $B_{\mathrm{tr}}$. Henceforth, we denote $\phi_k = \begin{bmatrix}x_k^\top & u_k^\top \end{bmatrix}^\top \in \mathbb{R}^{n_\mathrm{\phi}}$ where $n_\mathrm{\phi}=n_\mathrm{x} + n_\mathrm{u}$. In our setting, we assume that some prior knowledge on the dynamics is available.
\begin{assumption} \label{a1} The parameters $\theta=\mathrm{vec}(A,B)$ have a Gaussian prior, i.e., $\theta \sim  \mathcal{N}(\hat{\theta}_{\mathrm{prior}},\Sigma_{\theta, \mathrm{prior}})$. Furthermore, there exists a matrix $\tilde{D}_0 \succ 0$ such that $\Sigma_{\theta, \mathrm{prior}}^{-1}=\tilde{D}_0 \otimes I_{n_\mathrm{x}}$.
\end{assumption} 

The estimate $\hat{\theta}_T=\text{vec}([\hat{A}_T,\hat{B}_T])$ is the maximum a posteriori (MAP) estimate of the unknown parameters $A_{\mathrm{tr}}$ and $B_{\mathrm{tr}}$, and is computed as:

\begin{equation}
\begin{split}
\label{eq:LMS}
\hat{\theta}_T= \argmin_{\theta} \sum_{k=0}^{T-1}& \tfrac{1}{\sigma_\mathrm{w}^2} \left\|( x_{k+1}-\left(\left[x_k^\top \; u_k^\top\right]\otimes I_{n_\mathrm{x}}\right) \theta)\right\|^2\\
&+\left\|\theta-\hat{\theta}_{\mathrm{prior}}\right\|_{\Sigma_{\theta,\mathrm{prior}}^{-1}}^2
\end{split}
\end{equation}
with posterior covariance
\begin{equation}\label{eq:post_cov}
\Sigma_{\theta,\mathrm{post}}^{-1}=\Sigma_{\theta, \mathrm{prior}}^{-1}+\left(\tfrac{1}{\sigma_\mathrm{w}^2}\sum_{k=0}^{T-1}\phi_k\phi_k^\top\right) \otimes I_{n_\mathrm{x}}.
\end{equation}
Under Assumption \ref{a1}, the posterior distribution $p(\theta|\mathcal{D})$ is given by $\mathcal{N}(\hat{\theta}_T,\Sigma_{\theta,\mathrm{post}})$ \cite[Prop. 2.1]{umenberger2019robust}.

\begin{remark}\label{rem:rem1} In case a prior as in Assumption \ref{a1} is not available, it can also be inferred from data. More precisely, given a data set $\mathcal{D}_0=\{\phi_k \}_{k=-\bar{T}}^{-1}$ obtained from a randomly exciting input, and a uniform prior over the parameters $\theta=\mathrm{vec}([A,B])$, i.e., $p(\theta)~\propto~1$, the posterior distribution is given by $\mathcal{N}(\hat{\theta},\Sigma_\theta)$, where $\hat{\theta}=\mathrm{vec}([\hat{A},\hat{B}])$ is the ordinary least squares estimate, and $\Sigma_\theta^{-1}=\left(\frac{1}{\sigma_\mathrm{w}^2} \sum_{k=-\bar{T}}^{-1}\phi_k \phi_k^\top\right) \otimes I_{n_\mathrm{x}}=\tilde{D}_0 \otimes I_{n_\mathrm{x}}.$ This also justifies the structural assumption on $\Sigma_{\theta, \mathrm{prior}}^{-1}$ being of the form $\tilde{D}_0 \otimes I_{n_\mathrm{x}}$ for some $\tilde{D}_0 \succ 0$.
\end{remark}

 The following lemma provides high-probability credibility regions for the uncertain parameters $\theta$, and the uncertain system matrices $A$, $B$. Throughout this paper, we make the standing assumption that the length $T\rightarrow\infty$, such that the following asymptotic identification bounds can be applied.

\begin{lemma}\cite[Prop. 2.1, Lem. 3.1]{umenberger2019robust} \label{lem1}Let Assumption~\ref{a1} hold. Given data set  $\mathcal{D}_T$ of length $T$ with estimate $\hat{\theta}_T=\text{vec}([\hat{A}_T,\hat{B}_T])$ (cf. \eqref{eq:LMS}), set $D_{T}= \frac{1}{c_\delta \sigma_\mathrm{w}^2} \sum_{k=0}^{T-1}\phi_k \phi_k^\top$, $D_0=\frac{1}{c_\delta}\tilde{D}_0$,  $c_\delta=\chi_{n_\mathrm{x} n_\mathrm{\phi}}^{2}(\delta)$ with $0~<~\delta~<~1$. Then,
\begin{itemize}
\item[{\textbf{I.}}] $\mathbb{P}(\theta_{\mathrm{tr}}=\mathrm{vec}([A_{\mathrm{tr}},B_{\mathrm{tr}}]) \in \mathbf{\Theta}) = 1-\delta$, where
\begin{equation}\label{eq:lem1_D}
\mathbf{\Theta}:=\left\{ \theta:(\theta-\hat{\theta}_T)^\top \Sigma_{\theta,\mathrm{post}}^{-1} (\theta-\hat{\theta}_T) \leq c_\delta \right\}
\end{equation}
with $\Sigma_{\theta,\mathrm{post}}^{-1}=\Sigma_{\theta,\mathrm{prior}}^{-1}+\left(c_\delta D_T \right)\otimes I_{n_\mathrm{x}} $, and

\item[{\textbf{II.}}] $\mathbb{P}([A_{\mathrm{tr}},B_{\mathrm{tr}}] \in \mathbf{\Delta}) = 1-\delta$, where
\begin{small}
\begin{equation}\label{credregion}
\mathbf{\Delta}:=\left\{A,B:\begin{bmatrix}
(\hat{A}_T-A)^\top\\(\hat{B}_T-B)^\top
\end{bmatrix}^\top D_\mathrm{post} \begin{bmatrix}
(\hat{A}_T-A)^\top\\(\hat{B}_T-B)^\top
\end{bmatrix} \preceq I\right\}
\end{equation}
\end{small}
with $D_\mathrm{post}=D_0+D_T$ .
\end{itemize}
\end{lemma}

 The result of Lemma~\ref{lem1} is a data-dependent uncertainty bound that can be utilized to synthesize robust controllers similar to approaches in \cite{umenberger2019robust, ferizbegovic2019learning, iannelli2019structured, venkatasubramanian2020robust}.
 Note that statement $\textbf{I}$ implies statement $\textbf{II}$ of Lemma \ref{lem1} \cite{umenberger2019robust}.  In Section \ref{sec:exploration}, we use this result to derive a suitable targeted exploration strategy.

%

%
Initial estimates of the system parameters can be obtained from the mean of the prior distribution via $\text{vec}([\hat{A}_0,\hat{B}_0])=\hat{\theta}_{\mathrm{prior}}$. The matrix $D_0$ quantifies the robust bound associated with these initial estimates for a given probability $1-\delta$. More precisely, from Lemma \ref{lem1}, $\theta_\mathrm{tr} \in \mathbf{\Theta}_0$ with probability $1-\delta$, where
\begin{align}\label{eq:Theta0}
\mathbf{\Theta}_0:=\left\{\theta:(\hat{\theta}_{\mathrm{prior}}-\theta)^\top (D_0 \otimes I_{n_\mathrm{x}}) (\hat{\theta}_{\mathrm{prior}}-\theta) \leq 1 \right\},
\end{align}
and $[A_\mathrm{tr},B_\mathrm{tr}] \in \mathbf{\Delta}_0$ with probability $1-\delta$ where
\begin{equation}\label{eq:Delta0}
\mathbf{\Delta}_0:=\left\{A,B:\begin{bmatrix}
(\hat{A}_0-A)^\top\\(\hat{B}_0-B)^\top
\end{bmatrix}^\top D_0 \begin{bmatrix}
(\hat{A}_0-A)^\top\\(\hat{B}_0-B)^\top
\end{bmatrix} \preceq I\right\}.
\end{equation}
Denote 
\begin{align}\label{eq:prior_ub}
\Delta_0=\begin{bmatrix}
A_\mathrm{tr}-\hat{A}_0 & B_\mathrm{tr} - \hat{B}_0
\end{bmatrix}.
\end{align}
By applying the Schur complement twice on \eqref{eq:Delta0}, Lemma \ref{lem1} implies that
\begin{align}\label{eq:prior_ub_prob}
\mathbb{P}(\Delta_0^\top \Delta_0 \preceq D_0^{-1}) = 1-\delta.
\end{align}
 
Through the targeted exploration process for $T$ time steps as elaborated in Section \ref{sec:exploration}, data $\mathcal{D}_T=\{\phi_k\}_{k=0}^{T-1}$ will be observed. The new estimates $\hat{A}_T$ and $\hat{B}_T$ will be computed from data $\mathcal{D}_T$ and the prior information as in \eqref{eq:LMS}, and made available at time $T$. The matrix $D_\mathrm{post} \coloneqq D_0+\frac{1}{c_\delta \sigma_\mathrm{w}^2}\sum_{k=0}^{T-1}\phi_k \phi_k^\top$ quantifies the uncertainty associated with the estimates $\hat{A}_T$ and $\hat{B}_T$ (cf. Lemma \ref{lem1}). 

\subsection{Frequency domain information using spectral lines}
\label{sec:prelim_spectral}
This subsection discusses preliminaries for finite excitation based on the theory of spectral lines, which deals with the analysis of frequency domain information that can be derived from time-series data \cite{sarker2020parameter}. The data-dependent uncertainty bounds in Lemma~\ref{lem1} are based on the matrix $D_{T}$, which is a quantitative measure of finite excitation \cite[Def. 3]{sarker2020parameter}. The following notion of a sub-Gaussian spectral line is introduced.

\begin{definition}{(Sub-Gaussian Spectral Line \cite{sarker2020parameter})} \label{def:subg_spectral line}A stochastic sequence $\{\phi_k\}_{k=0}^{T-1}$ is said to have a sub-Gaussian spectral line from $k=0$ to $T-1$ at a frequency $\omega_0 \in \Omega_T$ with amplitude $\bar{\phi}(\omega_0) \in \mathbb{C}^{n_\phi}$ and a radius $R>0$ if
\begin{equation}
\frac{1}{T}\sum_{k=0}^{T-1}\phi_k e^{-j2\pi \omega_0 k} - \bar{\phi}(\omega_0) \sim \mathsf{subG}(R^2/T).
\end{equation}
\end{definition}

If noise is neglected, we recover a deterministic frequency component
\begin{equation}
\bar{\phi}(\omega_0)=\frac{1}{T}\sum_{k=0}^{T-1}\phi_k e^{-j2\pi \omega_0 k}.
\end{equation}
The radius $R$ in Definition \ref{def:subg_spectral line} is a measure of the stochastic process noise. The notion of a sub-Gaussian spectral line induces a requirement of appropriate behaviour over finite time. If the input has sufficiently many spectral lines, then the input signal is finitely exciting and can be used to provide bounds for parameter estimation. In order to establish the relationship between the spectral content of the input signal and finite excitation, the expected information matrix is first defined as follows.

\begin{definition}{(Expected Information Matrix \cite{sarker2020parameter})} \label{def:info_matrix} Given a sequence $\{\phi_k\}_{k=0}^{T-1}$ with $L$ sub-Gaussian spectral lines at frequencies $\omega_i \in \Omega_T,\,i=1,\dots,L$ with amplitudes
$\{\bar{\phi}(\omega_i)\}_{i=1}^{L}$, the information matrix $\bar{\Phi} \in \mathbb{C}^{n_\phi \times L}$ is defined as
\begin{equation}\label{eq:info_matrix}
\bar{\Phi}=\begin{bmatrix}
\vrule & \ldots & \vrule \\
\bar{\phi}(\omega_1) & \ldots & \bar{\phi}(\omega_L)\\
\vrule & \ldots & \vrule
\end{bmatrix}.
\end{equation}
\end{definition}

The spectral content in the expected information matrix can be used to determine whether a signal is finitely exciting or not \cite{sarker2020parameter}. In deterministic system identification, estimation of unknown parameters is made possible if $\bar{\Phi}$ has full rank and is numerically well conditioned. Since $\omega_i \in \Omega_T, \, i=1,\dots,L$, we set $T \geq L$ and select $L$ frequencies from $\Omega_T$.

\section{Targeted Exploration}\label{sec:exploration}
In this section, we propose a targeted exploration strategy based on a derived uncertainty bound on the data obtained through the process of exploration. Unlike greedy random exploration \cite{umenberger2019robust, ferizbegovic2019learning, iannelli2019structured, venkatasubramanian2020robust}, we introduce a targeted exploration strategy in the form of a linear combination of sinusoids with specified frequencies that explicitly shape the model uncertainty. In particular, this is achieved by deriving a lower bound on finite excitation of the exploration data using the spectral information of the exploration inputs in Lemma \ref{lemma:uncertainty bound spectral}. Since this bound is non-convex in the decision variables and depends on uncertain model parameters, a convex relaxation procedure is carried out (cf. Section \ref{subsec:convex_rel}), and required bounds on the effect of the model uncertainty are derived (cf. Section \ref{subsec:hinf_bound}). Finally, in Section \ref{subsec:exp_lowerbound}, we obtain an LMI for exploration which provides us with exploration inputs that guarantee a lower-bound on the finite excitation. In Section \ref{sec:dual_control}, the results of this section will be combined  with a robust (gain-scheduled) control design to achieve a guaranteed performance after exploration.

\subsection{Exploration strategy}
The exploration input sequence takes the form
\begin{align}\label{eq:exploration_controller}
u_k=\sum_{i=1}^{L} a_i \cos(2 \pi \omega_i k), \quad\, k=0,\dots,T-1
\end{align}
where $T$ is the exploration time and $a_i \in \mathbb{R}^{n_\mathrm{u} \times 1}$ are the amplitudes of the sinusoidal inputs at $L$ distinctly selected frequencies $\omega_i \in \Omega_T$. Since the input signal is deterministic and sinusoidal, the amplitude of the spectral line is $\overline{u}(\omega_i)=a_i$, and the radius of the spectral line is $0$ (cf. Def. \ref{def:subg_spectral line}).
Denote $U_\mathrm{e}=\mathrm{diag}(a_1,\dots,a_{L}) \in \mathbb{R}^{Ln_\mathrm{u} \times L}$. The exploration input is computed such that it excites the system sufficiently with minimal control energy, based on the initial parameter estimates. To this end, we require that the control energy at each time instant does not exceed $\gamma^2_\mathrm{e}$, i.e., $\sum_{i=1}^L \|a_i\|^2= \mathbf{1}_{L}^\top U_\mathrm{e}^\top U_\mathrm{e} \mathbf{1}_{L} \preceq \gamma_\mathrm{e}^2$ where $\mathbf{1}_{L} \in \mathbb{R}^{L\times 1}$ is a vector of ones, and the bound $\gamma_\mathrm{e} \geq 0$ is desired to be small. Using the Schur complement, this criterion is equivalent to 
\begin{equation}\label{eq:min_energy_cost}
S_{\textnormal{energy-bound}}(\gamma_\mathrm{e},U_\mathrm{e})\coloneqq \begin{bmatrix} 
\gamma_\mathrm{e} & \mathbf{1}_{L}^\top U_\mathrm{e}^\top \\ U_\mathrm{e} \mathbf{1}_{L} & \gamma_\mathrm{e} I
\end{bmatrix} \succeq 0.
\end{equation}

Since we consider only open-loop inputs in our exploration strategy, we require $A_\mathrm{tr}$ to be Schur stable.

\begin{remark} \label{rem:2}An exploration input of the form in \eqref{eq:exploration_controller} with an additional linear feedback, i.e., $v_k=u_k+Kx_k$, which robustly stabilizes the initial estimate and prior uncertainty, may be considered if it is not known whether the system is Schur stable.
\end{remark}

\subsection{Bound on finite excitation}\label{sec:param_bounds}
For the system evolving under the exploration input as given in (\ref{eq:exploration_controller}), the uncertainty bound on the parameters can be computed from the expected information matrix $\bar{\Phi}$ of the input. As a prerequisite to computing the uncertainty bound, it is necessary to establish the relationship between the spectral content of the observed state $x_k$ and the applied input $u_k$. At a frequency $\omega_i$, the relation between $\phi_k$, $u_k$ and $w_k$ is

\begin{equation}\label{eq:tf}
\begin{split}
\bm{\phi}(e^{j \omega_i}) = &\underbrace{\begin{bmatrix}
(e^{j \omega_i}I - A_\mathrm{tr})^{-1}B_\mathrm{tr} \\ I_{n_\mathrm{u}}
\end{bmatrix}}_{V_i}\mathbf{u}(e^{j \omega_i})\\
& + \underbrace{\begin{bmatrix}
(e^{j \omega_i}I - A_\mathrm{tr})^{-1} \\ 0
\end{bmatrix}}_{Y_i}\mathbf{w}(e^{j \omega_i}).
\end{split}
\end{equation}

Given $u_k$ as in \eqref{eq:exploration_controller} and using \cite[Lemma 1]{sarker2020parameter},
$\phi_k$ has $L$ sub-Gaussian spectral lines (cf. Def. \ref{def:subg_spectral line}) from $0$ to $T-1$ at distinct frequencies $\omega_i \in \Omega_T$, $i=1,\dots,L$ with amplitudes
\begin{equation}
\bar{\phi}(\omega_i)=V_i \bar{u}(\omega_i) 
\end{equation}
and radii $\left\|Y_i \right\|\sigma_\mathrm{w}$. 
Denote 
\begin{align}
V_\mathrm{tr}=[V_1,\cdots,V_{L}] \in \mathbb{C}^{n_\mathrm{\phi}\times L n_\mathrm{u}},
\end{align}
which is unknown since the true dynamics $A_\mathrm{tr}$, $B_\mathrm{tr}$ are unknown. Then, the expected information matrix $\bar{\Phi} \in \mathbb{C}^{n_\mathrm{\phi} \times L}$(cf. Def. \ref{def:info_matrix}) is
\begin{equation}\label{eq:inf_matrix}
\bar{\Phi}=V_\mathrm{tr}U_\mathrm{e}.
\end{equation}
Denote
\begin{align}
\nonumber
Y_\mathrm{tr}&=[Y_1,\cdots,Y_{L}]\in \mathbb{C}^{n_\phi \times L n_\mathrm{x}},\\
W&=\mathrm{diag}(\bar{w}(\omega_1),\dots,\bar{w}(\omega_L)) \in \mathbb{C}^{L n_\mathrm{x}\times L} 
\end{align}
where $\bar{w}(\omega_i)=\frac{\mathbf{w}(e^{j \omega_i})}{T}$, and
\begin{equation}\label{eq:wtilde}
\tilde{W}=Y_\mathrm{tr} W.
\end{equation}
Note that
\begin{align}\label{eq:clar_phibar}
    \frac{1}{T}\bm{\phi}(e^{j \omega_i})=\bar{\phi}(\omega_i)+Y_i \bar{w}(\omega_i).
\end{align}
The effect of the noise $w$ is captured by $W$. Each block $\bar{w}(\omega_i),\,i=1,...,L,$ on the diagonal of $W$ is Gaussian, i.e., $\bar{w}(\omega_i) \sim \mathcal{N}(0,\sigma_{\bar{\mathrm{w}}}^2 I)$ with $\sigma_{\bar{\mathrm{w}}}^2=\frac{\sigma_\mathrm{w}^2}{T}$, since the DFT of a Gaussian signal is also Gaussian, but with a distinct variance (cf. Appendix \ref{appendix:noise_radius}).
The following lemma, inspired by \cite[Proposition 3]{sarker2020parameter}\footnote{The penultimate step in the proof of \cite[Prop. 3, Appendix D]{sarker2020parameter} results in $\|\bar{\Phi}+\tilde{W}\|^{2} \geq \| \bar{\Phi}^{-1}\|^{-2} - \| \tilde{W}\|^2$, which is incorrect, in general. The authors of \cite{sarker2020parameter} have prepared a corrigendum \cite{sarker2023corrig} that avoids these arguments in the proof of \cite[Prop. 3, Appendix D]{sarker2020parameter}.}, presents a clear relationship between the spectral content of the signal and finite excitation.
\begin{lemma}{}\label{lemma:uncertainty bound spectral} For any $\epsilon \in (0,1)$, $\phi_k$ satisfies

\begin{equation}\label{eq:lem3}
\frac{1}{T}\sum_{k=0}^{T-1}\phi_k \phi_k^\top \succeq \frac{1}{L} \left(  (1-\epsilon) \bar{\Phi} \bar{\Phi}^\mathsf{H} + \left(\frac{\epsilon-1}{\epsilon}\right) \tilde{W} \tilde{W}^\mathsf{H}  \right).
\end{equation}
\end{lemma}

\begin{proof}
Note that for any unit vector $z \in \mathbb{C}^{n_\mathrm{\phi}}$, and any realization $\{ \phi_k \}_{k=0}^{T-1}$,
\begin{align}\label{eq:spectral_bound_1}
\nonumber
z^\mathsf{H} \left( \frac{1}{T}\sum_{k=0}^{T-1}\phi_k \phi_k^\top\right) z & = \frac{1}{T} \sum_{k=0}^{T-1} \left|\phi_k^\top z\right|^2 \\
\nonumber
& = \frac{1}{T} \sum_{k=0}^{T-1} \left| \phi^\top_k z \right|^2 \cdot \underbrace{\frac{1}{L} \sum_{l=1}^{L} \left|e^{-j2\pi \omega_l k} \right|^2}_{=1}\\
& \geq \frac{1}{L}\sum_{l=1}^{L} \left| \frac{1}{T}\sum_{k=0}^{T-1} \phi_k^\top z e^{-j2\pi \omega_l k} \right|^2,
\end{align}
by Jensen's inequality. This leads to
\begin{align}\label{eq:spectral_bound_2}
\nonumber
z^\mathsf{H} \left( \frac{1}{T}\sum_{k=0}^{T-1}\phi_k \phi_k^\top \right) z & \geq \frac{1}{L}\sum_{l=1}^{L} \left| \frac{1}{T}\sum_{k=0}^{T-1} \phi_k^\top z e^{-j2\pi \omega_l k} \right|^2 \\
& \overset{\eqref{eq:wtilde},\eqref{eq:clar_phibar}}{=} \frac{1}{L} \left( z^\mathsf{H} ( \bar{\Phi}+\tilde{W} )( \bar{\Phi}+\tilde{W} )^\mathsf{H} z \right).
\end{align}
By Young's inequality \cite{caverly2019lmi}, for any $\epsilon >0$, we have
\begin{equation}
\bar{\Phi} \tilde{W}^\mathsf{H} + \tilde{W} \bar{\Phi}^\mathsf{H} \succeq -\epsilon \bar{\Phi} \bar{\Phi}^\mathsf{H} - \frac{1}{\epsilon} \tilde{W} \tilde{W}^\mathsf{H}
\end{equation}
and hence
\begin{equation}\label{youngs_ineq}
(\bar{\Phi}+\tilde{W})(\bar{\Phi}+\tilde{W})^\mathsf{H} \succeq (1-\epsilon) \bar{\Phi} \bar{\Phi}^\mathsf{H} - \left(\frac{1-\epsilon}{\epsilon}\right) \tilde{W} \tilde{W}^\mathsf{H}.
\end{equation}
By inserting Inequality \eqref{youngs_ineq} in Inequality \eqref{eq:spectral_bound_2}, we get \eqref{eq:lem3}.

\end{proof}

From Lemma \ref{lemma:uncertainty bound spectral}, using \eqref{eq:inf_matrix} and \eqref{eq:lem3}, for any $\epsilon \in (0,1)$, $\phi_k$ satisfies

\begin{small}
\begin{align}\label{eq:lowerbound}
\underbrace{\frac{1}{\bar{c}}\sum_{k=0}^{T-1}\phi_k \phi_k^\top}_{D_T} \succeq \frac{T}{\bar{c}L}\Bigg((1-\epsilon)V_\mathrm{tr} U_\mathrm{e} U_\mathrm{e}^\top {V_\mathrm{tr}}^\mathsf{H} 
- \left(\frac{1-\epsilon}{\epsilon}\right) \tilde{W}\tilde{W}^\mathsf{H} \Bigg)
\end{align}
\end{small}
where $\bar{c}=\sigma_\mathrm{w}^2 c_\delta$.

Inequality \eqref{eq:lowerbound} allows us to determine a lower bound on the finite excitation $D_\mathrm{post}=D_0+D_T$, and thus upper-bound the uncertainty of the MAP estimate using Lemma \ref{lem1}. Note that this lower bound depends on the amplitudes of the harmonic signals $U_\mathrm{e}$ (cf. \eqref{eq:info_matrix}, \eqref{eq:inf_matrix}), as well as on the size of the noise. However, determining a lower bound on $D_T$ based on Inequality \eqref{eq:lowerbound} results in non-convex constraints in the decision variable $U_\mathrm{e}$. We circumvent this issue by using a convex relaxation procedure.

\begin{table*}
\begin{small}
   \begin{align}\label{eq:explorationLMI}
S_{\textnormal{exploration}}(\epsilon,\tau,\tilde{U},U_\mathrm{e},\hat{V},l,\overline{D}_T,\Gamma_\mathrm{v}) = \begin{bmatrix}
        \left(1-\epsilon\right) (U_\mathrm{e}^\top \tilde{U} +  \tilde{U}^\top U_\mathrm{e} - \tilde{U}^\top \tilde{U}) & 0\\0 & -\left(\frac{1-\epsilon}{\epsilon}\right)l^2 I - \frac{\bar{c}\, L}{T} \overline{D}_T
    \end{bmatrix} -\tau \begin{bmatrix}
        -I & \hat{V}^\mathsf{H}\\ \hat{V} & \Gamma_\mathrm{v}-\hat{V}\hat{V}^\mathsf{H}
    \end{bmatrix}  \succeq 0    
\end{align} 
\end{small}
\end{table*}

\subsection{Convex relaxation}\label{subsec:convex_rel}
The following lemma provides a lower bound on $D_T$ which is linear in $U_\mathrm{e}$.
\begin{lemma}\label{lem:convexrel} For any matrices $\tilde{U} \in \mathbb{R}^{Ln_\mathrm{u} \times L}$ and $U_\mathrm{e} \in \mathbb{R}^{Ln_\mathrm{u} \times L}$, and any $\epsilon \in (0,1)$, we have:
\begin{align}\label{eq:convex_relaxation}
\nonumber
\frac{\bar{c}L}{T} D_T  \succeq & \left(1-\epsilon\right) \left( V_\mathrm{tr}\left( U_\mathrm{e} \tilde{U}^\top +  \tilde{U} U_\mathrm{e}^\top - \tilde{U} \tilde{U}^\top \right)V_\mathrm{tr}^\mathsf{H}\right)\\
& - \left(\frac{1-\epsilon}{\epsilon}\right)\tilde{W} \tilde{W}^\mathsf{H} .
\end{align} 
\end{lemma}
\begin{proof}
We have 
\begin{align}
\nonumber
&V_\mathrm{tr} U_\mathrm{e} U_\mathrm{e}^\top {V_\mathrm{tr}}^\mathsf{H} - V_\mathrm{tr} U_\mathrm{e} \tilde{U}^\top V_\mathrm{tr}^\mathsf{H} - V_\mathrm{tr} \tilde{U} U_\mathrm{e}^\top V_\mathrm{tr}^\mathsf{H} + V_\mathrm{tr} \tilde{U} \tilde{U}^\top V_\mathrm{tr}^\mathsf{H}\\
\nonumber
&=V_\mathrm{tr} (U_\mathrm{e}-\tilde{U})(U_\mathrm{e}-\tilde{U})^\top V_\mathrm{tr}^\mathsf{H}\\
\nonumber
& \succeq 0
\end{align}
and hence
\begin{align}\label{eq:conv_rel_1}
V_\mathrm{tr} U_\mathrm{e} U_\mathrm{e}^\top {V_\mathrm{tr}}^\mathsf{H}  \succeq V_\mathrm{tr}\left( U_\mathrm{e} \tilde{U}^\top +  \tilde{U} U_\mathrm{e}^\top - \tilde{U} \tilde{U}^\top \right)V_\mathrm{tr}^\mathsf{H}.
\end{align}
Inserting Inequality \eqref{eq:conv_rel_1} in Inequality \eqref{eq:lowerbound} leads to \eqref{eq:convex_relaxation}.\qedhere
\end{proof}

The bound derived in Lemma \ref{lem:convexrel} is tight in case $\tilde{U}=U_\mathrm{e}$. However, since $U_\mathrm{e}$ is unknown, we consider a candidate $\tilde{U}$ of linearly independent amplitudes corresponding to their respective spectral lines. We later embed this relaxation in an iterative process to reduce conservatism. Furthermore,  in \eqref{eq:convex_relaxation}, $V_\mathrm{tr}$ and $Y_\mathrm{tr}$ (in $\tilde{W}=Y_\mathrm{tr}W$) are unknown. Hence, in what follows, suitable bounds are derived.

\subsection{Bounds on transfer matrices}\label{subsec:hinf_bound}
Denote 
\begin{align}\label{eq:vtilde}
    \tilde{V}=V_\mathrm{tr}-\hat{V}
\end{align}
where the estimate
\begin{align}\label{eq:Vhat}
\hat{V}=[\hat{V}_1,\cdots,\hat{V}_{L}] \in \mathbb{C}^{n_\mathrm{\phi}\times Ln_\mathrm{u}}
\end{align} is computed using the prior $\hat{A}_0$ and $\hat{B}_0$ (cf. Assumption \ref{a1}). In Appendices \ref{appendix:gamma}, \ref{appendix:gamma_y} and \ref{appendix:sample}, we show how to compute a matrix $\Gamma_\mathrm{v} \succ 0$ and a constant $\gamma_\mathrm{y} >0$ such that
\begin{align}\label{eq:tf_prop}
\tilde{V}\tilde{V}^\mathsf{H} \preceq \Gamma_\mathrm{v},\;\|Y_\mathrm{tr}\| \leq \gamma_\mathrm{y},
\end{align}
assuming $\theta_\mathrm{tr} \in \mathbf{\Theta}_0$.

Utilizing \eqref{eq:tf_prop}, we can derive a bound on $\tilde{W}$ of the form $\|\tilde{W}\| \leq \|Y_\mathrm{tr}\| \|W\|$. A bound on $W$ can be determined since the each block on the diagonal of $W$ is Gaussian. In particular, we have
\begin{align}
\|W\| = \max_{i=1,\dots,L}\| \bar{w}(\omega_i) \|.
\end{align}
Since $\bar{w}(\omega_i) \sim \mathcal{N}(0,\sigma_{\bar{\mathrm{w}}}^2 I)$, we have that $\| \bar{w}(\omega_i) \|^2 \sim \sigma_{\bar{\mathrm{w}}}^2\chi_{n_\mathrm{x}}^2$, and hence
\begin{align}\label{eq:l1}
\mathbb{P}\left( \| \bar{w}(\omega_i)  \|^2 \leq l_1^2 \right) = 1-\delta,\, \forall i=1,...,L
\end{align}
with $l_1^2=\sigma_{\bar{\mathrm{w}}}^2\chi_{n_\mathrm{x}}^{2}(1-\delta)$. Therefore, we have
\begin{align}\label{eq:w_bound}
\mathbb{P}(\|W\| \leq l_1) = 1-\delta.
\end{align}
Using \eqref{eq:tf_prop} and \eqref{eq:w_bound}, we have
\begin{align}\label{eq:tildew_bound}
\tilde{W} \tilde{W}^\mathsf{H} \preceq l^2 I := (\gamma_\mathrm{y} l_1)^2 I.
\end{align}


The following lemma provides joint probabilistic bounds on $\tilde{V}$, $Y_\mathrm{tr}$ and $W$.

\begin{lemma}\label{lem:tfbounds}
Let Assumption \ref{a1} hold. Then
\begin{align}\label{eq:tfbound_prob}
\mathbb{P}( (\tilde{V}\tilde{V}^\mathsf{H} \preceq \Gamma_\mathrm{v})\cap(\|Y_\mathrm{tr}\|\leq \gamma_\mathrm{y})\cap(\|W\| \leq l_1)) \geq 1 -2\delta.
\end{align}
\end{lemma}
\begin{proof} From Lemma \ref{lem1}, since $\mathbb{P}(\theta_\mathrm{tr} \in \mathbf{\Theta}_0) = 1-\delta$, the bounds in \eqref{eq:tf_prop} hold with probability $1-\delta$:
\begin{align*}
\mathbb{P}( (\tilde{V}\tilde{V}^\mathsf{H} \preceq \Gamma_\mathrm{v}) \cap  (\|Y_\mathrm{tr}\| \leq \gamma_\mathrm{y} ) )& \geq \mathbb{P}(\theta_\mathrm{tr} \in \mathbf{\Theta}_0)\\
& = 1-\delta.
\end{align*}
Hence,
\begin{align}
\nonumber
&\mathbb{P}((\tilde{V}\tilde{V}^\mathsf{H} \preceq \Gamma_\mathrm{v})\cap(\|Y_\mathrm{tr}\|\leq \gamma_\mathrm{y})\cap(\|W\| \leq l_1))\\
\nonumber
& \geq \mathbb{P}((\theta_\mathrm{tr}\in \mathbf{\Theta}_0)\cap(\|W\| \leq l_1))\\
\nonumber
& \geq 1-\mathbb{P}(\theta_\mathrm{tr}\notin \mathbf{\Theta}_0)-\mathbb{P}(\|W\| \nleq l_1)\\
\label{eq:all_bounds}
& \overset{\eqref{eq:w_bound}}{\geq} 1-2\delta
\end{align}
wherein the penultimate inequality follows from De Morgan's law. \qedhere
\end{proof}

\subsection{Final bound on the informativity of exploration}\label{subsec:exp_lowerbound}
In the following theorem, we compute a lower bound $\overline{D}_T$ on the informativity, i.e., the empirical covariance of the exploration data, before the process of exploration.

\begin{theorem}\label{thm1}
Let Assumption \ref{a1} hold. Suppose there exist matrices $U_\mathrm{e}$ and $\overline{D}_T$, and $\tau \geq 0$ such that $S_{\textnormal{exploration}}(\epsilon,\tau,\tilde{U},U_\mathrm{e},\hat{V},l,\overline{D}_T,\Gamma_\mathrm{v}) \succeq 0$ as in \eqref{eq:explorationLMI}.
Then, the application of the input \eqref{eq:exploration_controller} implies that with probability at least $1-2\delta$:
\begin{align}\label{eq:thm1}
D_T\succeq \overline{D}_T.
\end{align}
\end{theorem}
\begin{proof} By using \eqref{eq:vtilde} and \eqref{eq:tf_prop}, we get
\begin{align*}
    V_\mathrm{tr} V_\mathrm{tr}^\mathsf{H}-V_\mathrm{tr} \hat{V}^\mathsf{H} - \hat{V} V_\mathrm{tr}^\mathsf{H} + \hat{V} \hat{V}^\mathsf{H} \preceq \Gamma_\mathrm{v}
\end{align*} 
which can be written as
\begin{align}\label{eq:Vtr_ineq}
    \begin{bmatrix}
        V_\mathrm{tr}^\mathsf{H} \\ I   \end{bmatrix}^\mathsf{H}\begin{bmatrix}
        -I & \hat{V}^\mathsf{H}\\ \hat{V} & \Gamma_\mathrm{v} -\hat{V}\hat{V}^\mathsf{H}
    \end{bmatrix} \begin{bmatrix}
        V_\mathrm{tr}^\mathsf{H} \\ I
    \end{bmatrix} \succeq 0.
\end{align}
Additionally, we have $\tilde{W} \tilde{W}^\mathsf{H} \preceq l^2 I$ as in \eqref{eq:tildew_bound}.

From Lemma \ref{lem:convexrel}, and by using \eqref{eq:tildew_bound}, the following inequality implies that $D_T \succeq \overline{D}_T$:
\begin{align}\label{eq:lowerbound_pre}
\nonumber
     \left(1-\epsilon\right) \left( V_\mathrm{tr}\left( U_\mathrm{e} \tilde{U}^\top +  \tilde{U} U_\mathrm{e}^\top
    - \tilde{U} \tilde{U}^\top \right)V_\mathrm{tr}^\mathsf{H}\right)&\\
    - \left(\frac{1-\epsilon}{\epsilon}\right)l^2 I -\frac{\bar{c}L}{T}\overline{D}_T & \succeq 0.
\end{align}
Furthermore, 
\eqref{eq:lowerbound_pre} can be written as
\begin{footnotesize}
\begin{align}
\nonumber
\begin{bmatrix}
       * \\ *
    \end{bmatrix}^\mathsf{H}\begin{bmatrix}
        (1-\epsilon) (U_\mathrm{e} \tilde{U}^\top +  \tilde{U} U_\mathrm{e}^\top - \tilde{U} \tilde{U}^\top) & 0\\0 & 
            -\big(\frac{1-\epsilon}{\epsilon}\big)l^2 I
            - \frac{\bar{c}\, L}{T} \overline{D}_T
    \end{bmatrix}\begin{bmatrix}
        V_\mathrm{tr}^\mathsf{H} \\ I
    \end{bmatrix}\\
    \succeq \quad 0.
\end{align}    
\end{footnotesize}
From Lemma \ref{lem:tfbounds}, Inequalities \eqref{eq:Vtr_ineq} and \eqref{eq:tildew_bound} hold jointly with probability $1-2\delta$.
By using the matrix S-lemma \cite{boyd2004convex,vanwaarde2022noisy}, \eqref{eq:lowerbound_pre} holds for all $V_\mathrm{tr}$ satisfying \eqref{eq:Vtr_ineq} if \eqref{eq:explorationLMI} holds with $\tau \geq 0$.

Hence, if there exist $U_\mathrm{e}$ and $\overline{D}_T$ satisfying \eqref{eq:explorationLMI}, then \eqref{eq:thm1} holds with probability at least $1-2\delta$.
\end{proof}

As a result, we can pose the exploration problem of achieving a desired excitation $\overline{D}_T$ with minimal input energy using the following SDP: \begin{align}
\nonumber
\underset{U_\mathrm{e},\gamma_\mathrm{e},\tau}{\inf}  & \quad \gamma_\mathrm{e} \\
\label{eq:exp_problem}
\text{s.t. }& \quad S_{\textnormal{energy-bound}}(\gamma_\mathrm{e},U_\mathrm{e})\succeq 0\\
\nonumber
& \quad S_{\textnormal{exploration}}(\epsilon,\tau,\tilde{U},U_\mathrm{e},\hat{V},l,\overline{D}_T,\Gamma_\mathrm{v}) \succeq 0\\
\nonumber
& \quad \tau \geq 0.
\end{align}

Note that $\gamma_\mathrm{e}$, $U_\mathrm{e}$ and $\tau$ are the only decision variables; $\epsilon$, $\tilde{U}$ and $\overline{D}_T$ are user-defined, and $\hat{V}$, $l$ and $\Gamma_{\mathrm{v}}$ are given by \eqref{eq:Vhat}, \eqref{eq:tildew_bound} and \eqref{eq:tf_prop}, respectively. A solution of \eqref{eq:exp_problem} gives us the parameters required for the implementation of the exploration input, i.e., $U_\mathrm{e}=\mathrm{diag}(a_1,\dots,a_{L})$, which guarantees the desired excitation $\overline{D}_T$. The lower-bound $\overline{D}_T$ also implies a lower bound $D_\mathrm{post} \succeq \overline{D}_\mathrm{post}=D_0+\overline{D}_T$, and will be used to guarantee an uncertainty bound for dual control in Section \ref{sec:dual_control}. In order to reduce the suboptimality caused by the convex relaxation procedure, Problem \eqref{eq:exp_problem} can be iterated multiple times by re-computing $\tilde{U}$ for the next iteration as
\begin{align}\label{eq:L_choice}
\tilde{U}=U_\mathrm{e}^*
\end{align}
wherein $U_\mathrm{e}^*$ is the solution from the previous iteration. The suboptimality of convex relaxation can be reduced by iterating until $U_\mathrm{e}$ does not change. Furthermore, $\gamma_\mathrm{e}$ is guaranteed to be non-increasing with each iteration since the previous optimal solution $U_\mathrm{e}^*$ remains feasible.


\section{Robust Gain-scheduling Design}\label{sec:gain_scheduling_h2}
After exploration as described in Section \ref{sec:exploration}, the estimates $\hat{A}$ and $\hat{B}$ change due to the availability of new data. Since our goal is to design a stabilizing state-feedback controller that is influenced by the new estimates, the true system in \eqref{sys} is modeled as a linear parameter varying (LPV) system, where the estimates $\hat{A}_T$, $\hat{B}_T$ are measured online. The new estimates are utilized as a \textit{scheduling variable} to design a gain-scheduling controller that ensures that the closed-loop system is stable while also satisfying a $H_2$ performance bound.

To account for the change in the estimates of the parameters after exploration, the system in (\ref{sys}) is rewritten as:
\begin{align}\label{gx}
x_{k+1}=&A_\mathrm{tr} x_k+B_\mathrm{tr} u_k+w_k\\\nonumber
=&\hat{A}_0x_k+\hat{B}_0u_k+(\hat{A}_T-\hat{A}_0)x_k+(\hat{B}_T-\hat{B}_0)u_k\\\nonumber
& + (A_\mathrm{tr}-\hat{A}_T)x_k +(B_\mathrm{tr}-\hat{B}_T)u_k +w_k.
\end{align}

Designing a controller for (\ref{gx}) can be written in the form of a gain-scheduling problem with an uncertain signal $w^\mathrm{u}~=~\Delta_\mathrm{u} \phi$ and an online measurable signal $w^\mathrm{s}=\Delta_\mathrm{s} \phi$ with the corresponding scheduling and uncertainty blocks:
\begin{equation}\label{eq:uncertainty_blocks}
\begin{split}
\Delta_\mathrm{s} & = \begin{bmatrix}
\hat{A}_T-\hat{A}_0 & \hat{B}_T-\hat{B}_0
\end{bmatrix}, \\
\Delta_\mathrm{u} & = \begin{bmatrix}
A_\mathrm{tr}-\hat{A}_T & B_\mathrm{tr}-\hat{B}_T
\end{bmatrix}.
\end{split}
\end{equation}

The control input after exploration is defined as
\begin{equation}\label{eq:input}
u_k=K_\mathrm{x} x_k +K_\mathrm{s} w_k^\mathrm{s},
\end{equation}
where we use the fact that $w_k^\mathrm{s}$ is known/measurable after exploration. The controller parameters $K_\mathrm{x}$ and $K_\mathrm{s}$ are designed such that the closed-loop system is robustly stable and the specified performance criterion is met. Since the estimates $\hat{A}_T$ and $\hat{B}_T$ affect both $\Delta_\mathrm{s}$ and $\Delta_\mathrm{u}$, these blocks can be viewed as uncertain parameters, and the closed-loop uncertain system combining (\ref{gz}) and (\ref{gx}) can be written in the generalized plant form:
\begin{align}\nonumber
&\begin{bmatrix}
x_{k+1}\\z_k^\mathrm{s} \\ z_k^\mathrm{u} \\z_k
\end{bmatrix}  = 
\begin{bmatrix}
\hat{A}_0 + \hat{B}_0 K_\mathrm{x} & I +\hat{B}_0 K_\mathrm{s} & I & I\\
\begin{bmatrix}
I\\K_\mathrm{x}
\end{bmatrix} & \begin{bmatrix}
0 \\ K_\mathrm{s}
\end{bmatrix} & 0 & 0 \\
\begin{bmatrix}
I\\K_\mathrm{x}
\end{bmatrix} & \begin{bmatrix}
0 \\ K_\mathrm{s}
\end{bmatrix} & 0 & 0 \\ 
C & 0 & 0 & 0
\end{bmatrix}
\begin{bmatrix}
x_k\\w_k^\mathrm{s}\\w_k^\mathrm{u}\\w_k
\end{bmatrix},\\\label{closedloop}
&\qquad\qquad\qquad\qquad\quad w_k^\mathrm{s} = \Delta_\mathrm{s} z_k^\mathrm{s},\\\nonumber
&\qquad\qquad\qquad\qquad\quad w_k^\mathrm{u} = \Delta_\mathrm{u} z_k^\mathrm{u},
\end{align}
where $w^\mathrm{s} \rightarrow z^\mathrm{s}$ is the scheduling channel and $w^\mathrm{u} \rightarrow z^\mathrm{u}$ is the uncertainty channel.

\begin{figure}[hbtp]
\begin{center}
\includegraphics[width=0.2\textwidth]{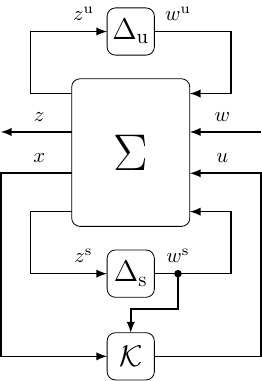}
\end{center}
\caption{Generalized plant view of the robust gain-scheduling problem \cite{venkatasubramanian2020robust}. }
\label{fig:GenPlant}
\end{figure}

As is apparent from Fig.~\ref{fig:GenPlant}, the closed-loop system has uncertainty and scheduling channels, affected by $\Delta_\mathrm{u}$ and $\Delta_\mathrm{s}$, respectively. The latter block $\Delta_\mathrm{s}$ is taken into account for the controller via $w_k^\mathrm{s}$ (cf.~\eqref{eq:input}), and hence plays the role of a scheduling variable.
This accounts for changes in the mean of the system parameters through data gathered in the exploration phase that is available after exploration at time $T$. 

Given this formulation and suitable bounds on the blocks $\Delta_\mathrm{s}$ and $\Delta_\mathrm{u}$, the following lemma provides a matrix inequality to design a robust gain scheduling controller satisfying the performance specification defined in \eqref{eq:h2perf}.

\begin{table*}[t]
\begin{subequations}\label{eq:LMIs_GS}
\begin{align}\label{eq:LMI_gs1}
S_{\textnormal{gain-scheduling-1}}(K_\mathrm{s},M,N,\lambda_\mathrm{s},\lambda_\mathrm{u},R_\mathrm{s}^{-1},R_{\mathrm{u}}^{-1})=\begin{pmatrix}
\begin{array}{c|c}
\begin{bmatrix}
-N & 0 & 0 & 0\\
0 & -\lambda_\mathrm{s} I & 0 & 0 \\
0 & 0 & -\lambda_\mathrm{u} I & 0 \\
0 & 0 & 0 & -\gamma_\mathrm{p}I
\end{bmatrix}&\star\\\hline
\begin{bmatrix}
\hat{A}_0N+\hat{B}_0 M& I+ \hat{B}_0 K_\mathrm{s}&I&I \\
\begin{bmatrix} N\\M \end{bmatrix} & \begin{bmatrix} 0\\K_\mathrm{s} \end{bmatrix} & 0 & 0 \\
 \begin{bmatrix} N\\M \end{bmatrix} & \begin{bmatrix} 0\\K_\mathrm{s} \end{bmatrix} & 0 & 0\\
0 & 0 & 0 & 0\\
\end{bmatrix}
&\begin{bmatrix}
-N & 0 & 0 & 0\\
0 & -\frac{1}{\lambda_\mathrm{s}}R_\mathrm{s}^{-1} & 0 & 0\\
0 & 0 & -\frac{1}{\lambda_\mathrm{u}}R_\mathrm{u}^{-1} & 0\\
0 & 0 & 0 & -I
\end{bmatrix}
\end{array}
\end{pmatrix}
\prec 0
\end{align}
\begin{align}\label{eq:LMI_gs2}
S_{\textnormal{gain-scheduling-2}}(M,N,Z)=\begin{bmatrix}
N & NC^\top \\ CN & Z
\end{bmatrix} \succ 0
\end{align}
\begin{align}\label{eq:LMI_gs3}
    S_{\textnormal{gain-scheduling-3}}(Z)=\text{trace}(Z)\leq \gamma_\mathrm{p}
\end{align}
\end{subequations}
\end{table*}

\begin{lemma}\label{lemma:robust}
Suppose $\Delta_\mathrm{s}\in\mathbf{\Delta}_\mathrm{s}:=\{\Delta: R_\mathrm{s}-\Delta^\top  \Delta \succ 0\}$, $\Delta_\mathrm{u}\in\mathbf{\Delta}_\mathrm{u}:=\{\Delta: R_\mathrm{u}-\Delta^\top \Delta  \succ 0\},$  with $R_\mathrm{u},R_\mathrm{s}\succ 0$.
If there exist matrices $K_\mathrm{s},M,N,Z$ and scalars $\lambda_\mathrm{s},\lambda_\mathrm{u}>0$ satisfying  the matrix inequalities \eqref{eq:LMIs_GS}, then the closed loop \eqref{closedloop} satisfies the $H_2$ performance bound $\gamma_\mathrm{p}$ \eqref{eq:h2perf} with $K_\mathrm{x}=MN^{-1}$, i.e., $u_k=MN^{-1}x_k+K_\mathrm{s} w_k^\mathrm{s}$. 
\end{lemma}

The proof of Lemma \ref{lemma:robust} is provided in Appendix \ref{appendix:robust_lemma}. We later define $R_\mathrm{s}$, $R_\mathrm{u}$ in terms of uncertainty bounds that affect the uncertain parameters $\Delta_\mathrm{s}$, $\Delta_\mathrm{u}$. For constants $\lambda_\mathrm{s},\lambda_\mathrm{u}$, Inequalities \eqref{eq:LMIs_GS} are LMIs. Hence, the matrix inequalities \eqref{eq:LMIs_GS} can be efficiently solved using line-search like techniques for $(\lambda_\mathrm{s},\lambda_\mathrm{u})\in\mathbb{R}^2$. 

The solution of \eqref{eq:LMIs_GS} yields control parameters $M$, $N$ and $K_\mathrm{s}$ which guarantee robust $H_2$ performance of the closed loop (\ref{closedloop}) for all bounded uncertainties $\Delta_\mathrm{s} \in \mathbf{\Delta}_\mathrm{s}$, $\Delta_\mathrm{u} \in \mathbf{\Delta}_\mathrm{u}$. In the next section, we combine the exploration strategy from Section \ref{sec:exploration}, and the parametrized state-feedback controller based on gain-scheduling in the overall dual control strategy.

\section{Dual Control}\label{sec:dual_control}
In this section, we propose a dual control strategy by combining targeted exploration (Section \ref{sec:exploration}), and robust gain-scheduling control (Section \ref{sec:gain_scheduling_h2}).

\subsection{Relationship between uncertain parameters}
An important aspect of the proposed approach lies in establishing the relationship between the uncertainty bounds $D_0$, $D_T$, that influence the uncertain parameters $\Delta_\mathrm{u}$, $\Delta_\mathrm{s}$ \eqref{eq:uncertainty_blocks}. Recall that $\theta_\mathrm{tr} \in \mathbf{\Theta}_0$ with high probability. In order to derive a bound on $\Delta_\mathrm{s}$ that is less conservative than the bound in \cite[Prop. 1]{venkatasubramanian2020robust}, we carry out the following parameter projection.

\subsubsection*{Parameter projection}
The estimate $\hat{\theta}_T=\textnormal{vec}([\hat{A}_T,\hat{B}_T])$ is projected on $\mathbf{\Theta}_0$ (cf. \eqref{eq:Theta0}) as follows:
\begin{equation}\label{eq:projection}
\tilde{\theta}_T=\Pi_{\mathbf{\Theta}_0}(\hat{\theta}_T):=\argmin_{\theta \in \mathbf{\Theta}_0} || \theta - \hat{\theta}_T ||_{(\overline{D}_\mathrm{post}\otimes I_{n_\mathrm{x}})}^2
\end{equation}
where $\overline{D}_\mathrm{post}=D_0+\overline{D}_T$ and $\tilde{\theta}_T=\mathrm{vec}([\tilde{A}_T,\tilde{B}_T])$ is the projected estimate. We redefine $\Delta_\mathrm{s}$, $\Delta_\mathrm{u}$ in terms of $(\tilde{A}_T,\tilde{B}_T)$ as
\begin{equation}\label{eq:Deltas_new}
\begin{split}
\Delta_\mathrm{s} & = \begin{bmatrix}
\tilde{A}_T-\hat{A}_0 & \tilde{B}_T-\hat{B}_0
\end{bmatrix}, \\
\Delta_\mathrm{u} & = \begin{bmatrix}
A_\mathrm{tr}-\tilde{A}_T & B_\mathrm{tr}-\tilde{B}_T
\end{bmatrix}.
\end{split}
\end{equation}
From Lemma \ref{lem1}, we have $\mathbb{P}(\theta_\mathrm{tr} \in \mathbf{\Theta}_T) \geq 1-\delta$ where
\begin{equation}\label{eq:hat_thetau}
\mathbf{\Theta}_T:=\left\{\theta:(\theta-\hat{\theta}_T)^\top (D_\mathrm{post} \otimes I_{n_\mathrm{x}}) (\theta-\hat{\theta}_T) \leq 1 \right\},
\end{equation}
with $D_\mathrm{post}=D_0+D_T$, and $\hat{\theta}_T$ is the estimate obtained after exploration \eqref{eq:LMS}. The following Lemma encapsulates the result of parameter projection and the bounds over uncertain parameters.

\begin{lemma}\label{lem:projection-bounds} Let Assumption \ref{a1} hold. If there exist matrices $U_\mathrm{e}$ and $\overline{D}_T$ satisfying the matrix inequality \eqref{eq:explorationLMI}, then  the following inequalities hold together with probability at least $1-3\delta$:
\begin{align}
\label{eq:delta0_bound}
\Delta_0^\top \Delta_0 & \preceq D_0^{-1},\\
\label{eq:deltas_bound}
\Delta_\mathrm{s}^\top \Delta_\mathrm{s} &\preceq D_0^{-1},\\
\label{eq:deltau_bound}
\Delta_\mathrm{u}^\top \Delta_\mathrm{u}  & \preceq \overline{D}_\mathrm{post}^{-1}.
\end{align} 

\end{lemma}

\begin{proof}The proof is provided in four parts. In the first three parts of the proof, we prove \eqref{eq:delta0_bound}-\eqref{eq:deltau_bound} by assuming $\theta_\mathrm{tr}~\in~ (\mathbf{\Theta}_0~\cap~\mathbf{\Theta}_T)$, where $\mathbf{\Theta}_T$ is defined in \eqref{eq:hat_thetau}, and $D_T \succeq \overline{D}_T$. In the last part, we then show that \eqref{eq:delta0_bound}-\eqref{eq:deltau_bound} hold jointly with probability at least $1-3\delta$.

\textbf{Part I.} Inequality \eqref{eq:delta0_bound} directly follows from $\theta_\mathrm{tr} \in \mathbf{\Theta}_0$ (cf.~\eqref{eq:Theta0} - \eqref{eq:prior_ub_prob}).

\textbf{Part II.} By definition of $\tilde{\theta}_T$ in \eqref{eq:projection}, we have $\tilde{\theta}_T \in \mathbf{\Theta}_0$, which implies $[\tilde{A}_T,\tilde{B}_T] \in \mathbf{\Delta}_0$ (cf. \eqref{eq:Theta0}-\eqref{eq:Delta0}), from Lemma \ref{lem1}. Hence Inequality \eqref{eq:deltas_bound} holds with $\Delta_\mathrm{s}$ according to \eqref{eq:Deltas_new}.

\textbf{Part III.} Let
\begin{align}
\mathbf{\overline{\Theta}}_T:=\left\{\theta:(\theta-\hat{\theta}_T)^\top (\overline{D}_\mathrm{post} \otimes I_{n_\mathrm{x}}) (\theta-\hat{\theta}_T) \leq 1 \right\}.
\end{align}
Note that $D_T \succeq \overline{D}_T$ implies $D_\mathrm{post} \succeq \overline{D}_\mathrm{post}$. If $D_\mathrm{post}~\succeq~\overline{D}_\mathrm{post}$, then $\mathbf{\Theta}_T \subseteq \overline{\mathbf{\Theta}}_T$. Combining this with $\theta_\mathrm{tr} \in \mathbf{\Theta}_T$ yields $\theta_\mathrm{tr} \in \mathbf{\overline{\Theta}}_T$.
Denote 
\begin{equation}
\mathbf{\Theta}_\mathrm{u}:=\left\{\theta:(\theta-\tilde{\theta}_T)^\top (\overline{D}_\mathrm{post} \otimes I_{n_\mathrm{x}}) (\theta-\tilde{\theta}_T) \leq 1 \right\}.
\end{equation}

Since $\theta_\mathrm{tr} \in \mathbf{\Theta}_0$, then we have
\begin{align*}
\|\theta_\mathrm{tr}-\tilde{\theta}_T\|_{(\overline{D}_\mathrm{post}\otimes I_{n_\mathrm{x}} )}\leq \|\theta_\mathrm{tr}-\hat{\theta}_T\|_{(\overline{D}_\mathrm{post}\otimes I_{n_\mathrm{x}})}
\end{align*}
given that the projection \eqref{eq:projection} is non-expansive (cf. Fig. \ref{fig:Delta}). Hence, we have $\theta_\mathrm{tr} \in \mathbf{\Theta}_\mathrm{u}$, which implies Inequality \eqref{eq:deltau_bound}.

\textbf{Part IV.} We have
\begin{align}\label{eq:finalbound}
\nonumber
&\mathbb{P}((\theta_\mathrm{tr} \in (\mathbf{\Theta}_0 \cap \mathbf{\Theta}_T))\cap (D_T \succeq \overline{D}_T))\\
\nonumber
\overset{\eqref{eq:all_bounds}}{\geq} & \mathbb{P}(\theta_\mathrm{tr} \in (\mathbf{\Theta}_0 \cap \mathbf{\Theta}_T))\cap(\|W\|\leq l_1))\\
\nonumber
\geq & 1- \mathbb{P}(\theta_\mathrm{tr} \notin \mathbf{\Theta}_0) - \mathbb{P}(\theta_\mathrm{tr} \notin \mathbf{\Theta}_T)- \mathbb{P}(\|W\| \nleq l_1)\\
\overset{\eqref{eq:hat_thetau}}{\geq}&  1-3\delta,
\end{align}
wherein the penultimate inequality follows from De Morgan's law. In Parts I-III, we showed that Inequalities \eqref{eq:delta0_bound}-\eqref{eq:deltau_bound} hold assuming $\theta_\mathrm{tr} \in (\mathbf{\Theta}_0 \cap \mathbf{\Theta}_T)$ and $D_T \succeq \overline{D}_T$. In \eqref{eq:finalbound}, we showed that $\theta_\mathrm{tr} \in (\mathbf{\Theta}_0 \cap \mathbf{\Theta}_T)$ and $D_T \succeq \overline{D}_T$ hold jointly with probability at least $1-3\delta$, which implies that \eqref{eq:delta0_bound}-\eqref{eq:deltau_bound} hold with probability at least $1-3\delta$. \qedhere
\end{proof}

\begin{figure}[h!]
\begin{center}
\includegraphics[width=0.35\textwidth]{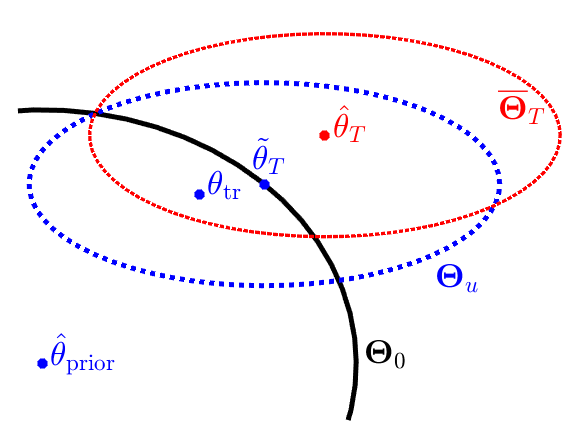}
\end{center}
\caption{Illustration of the sets from the proof of Lemma \ref{lem:projection-bounds}: $\mathbf{\Theta}_0=\mathbf{\Theta}_s,\text{ and }\mathbf{\Theta}_u$, the true parameters $\theta_\mathrm{tr}$, the initial parameter estimate $\hat{\theta}_\mathrm{prior}$, the estimate resulting from exploration $\hat{\theta}_T$, and the projected estimate $\tilde{\theta}_T$.}
\label{fig:Delta}
\end{figure}

The relationship between the different sets is illustrated in Figure \ref{fig:Delta}. From Lemma \ref{lem1}, it is known that the true system parameters $\theta_\mathrm{tr}$ are in some ellipse $\mathbf{\Theta}_0$ \eqref{eq:Theta0} around the initial parameter estimate $\hat{\theta}_\mathrm{prior}$. After exploration, $\hat{\theta}_T$ is obtained and projected to $\tilde{\theta}_T$ \eqref{eq:projection}. The projected $\tilde{\theta}_T$ is contained in the ellipse $\mathbf{\Theta}_0$. This ensures the following bound on $\Delta_\mathrm{s}$:~$\Delta_\mathrm{s}^\top \Delta_\mathrm{s} \preceq D_0^{-1}$ \eqref{eq:deltas_bound}. This bound on $\Delta_s$ is enabled by parameter projection and is less conservative than the bound proposed in \cite[Prop. 1]{venkatasubramanian2020robust}. The true system parameters $\theta_\mathrm{tr}$ are in some ellipse $\mathbf{\Theta}_\mathrm{u}$ around the projected estimate $\tilde{\theta}_T$ with some probability. Finally, from Lemma \ref{lem:projection-bounds}, $\theta_\mathrm{tr}$ lies in the intersection $\mathbf{\Theta}_0 \cap \mathbf{\Theta}_\mathrm{u}$ with some probability which ensures bounds on $\Delta_0$ and $\Delta_\mathrm{u}$ \eqref{eq:finalbound}. The satisfaction of these bounds enables the design of a robust gain-scheduled controller that guarantees $H_2$ performance \eqref{eq:LMIs_GS}, wherein we set $R_\mathrm{s}^{-1}=D_0$ and $R_\mathrm{u}^{-1}=\overline{D}_{\mathrm{post}}$.


\subsection{Proposed algorithm}
The objective of the proposed dual control approach is to ensure that the designed controller satisfies the $H_2$ performance bound \eqref{eq:h2perf} with high probability, by incorporating a suitable \textit{targeted} exploration strategy. 

The following SDP combines the robust gain-scheduling problem in \eqref{eq:LMIs_GS} and the exploration problem \eqref{eq:exp_problem}:

\begin{subequations}
\label{eq:opt_problem}
\begin{align}
&\underset{U_\mathrm{e},\overline{D}_T,K_\mathrm{s}, M, N, Z, \tau}{\inf}  \quad \gamma_\mathrm{e} \\\label{eq:opt_problem2}
\text{s.t. }& \quad S_{\textnormal{energy-bound}}(\gamma_\mathrm{e},U_\mathrm{e})\succeq 0\\
\label{eq:opt_problem_Ue}
& \quad S_{\textnormal{exploration}}(\epsilon,\tau,\tilde{U},U_\mathrm{e},\hat{V},l,\overline{D}_T,\Gamma_\mathrm{v}) \succeq 0\\
& \quad \tau \geq 0\\
\label{eq:opt_problem_gsc1}& \quad S_{\textnormal{gain-scheduling-1}}(K_\mathrm{s},M,N,\lambda_\mathrm{s},\lambda_\mathrm{u},R_\mathrm{s}^{-1},R_{\mathrm{u}}^{-1}) \prec 0\\
\label{eq:opt_problem_gsc2}& \quad S_{\textnormal{gain-scheduling-2}}(M,N,Z) \succ 0\\
\label{eq:opt_problem_gsc3}& \quad S_{\textnormal{gain-scheduling-3}}(Z) \leq \gamma_\mathrm{p}
\end{align}
\end{subequations}
Note that the SDP depends on the prior information $\hat{A}_0$, $\hat{B}_0$, and $D_0$ (cf. Assumption \ref{a1}), the desired $H_2$ performance specification $\gamma_\mathrm{p}$, the exploration time $T$, candidate $\tilde{U}$, and bounds $\Gamma_\mathrm{v}$, $\gamma_\mathrm{y}$, $l_1$. For fixed $\epsilon$, $\lambda_\mathrm{s}$, $\lambda_\mathrm{u}>0$, Problem \eqref{eq:opt_problem} is an SDP comparable to \cite[Eq. (28)]{venkatasubramanian2020robust}. Hence, the optimization problem can be solved efficiently using techniques like line-search or gridding for the variables $\epsilon$, $\lambda_\mathrm{s}$, $\lambda_\mathrm{u}>0$ with the SDP in an inner loop. Problem \eqref{eq:opt_problem} is feasible only if the desired performance specification $\gamma_\mathrm{p}$ is achievable by the gain-scheduling-based controller.
 
The solution of this optimization problem yields the exploration input $U_\mathrm{e}$, and parameters for the gain-scheduled controller $K_\mathrm{s}$, $M$ and $N$, which are subsequently used in the proposed sequential learning and control strategy. Recall that the exploration input is a linear combination of sinusoids, and can be implemented using the $L$ selected frequencies $\omega_i \in \Omega_T,\,i=1,...,L$, and their corresponding amplitudes $a_i=U_i$. The robust gain-scheduled controller can be implemented using $K_\mathrm{s}$ and $K_\mathrm{x}=MN^{-1}$. Furthermore, the robust control LMIs (cf.  Lemma~\ref{lemma:robust}) are contained in~\eqref{eq:opt_problem_gsc1}, \eqref{eq:opt_problem_gsc2} and \eqref{eq:opt_problem_gsc3}. These LMIs return a common Lyapunov function $N\succ0$ as well as controller parameters $M$, $K_\mathrm{s}$ which guarantee robust performance of the closed loop~\eqref{closedloop} for all uncertainties $\Delta_\mathrm{u},\Delta_\mathrm{s}$ satisfying $\Delta_\mathrm{u}^\top\Delta_\mathrm{u}\prec \overline{D}_\mathrm{post}^{-1}$ and $\Delta_\mathrm{s}^\top\Delta_\mathrm{s}\prec D_0^{-1}$. Within this framework, the data obtained during the exploration phase is lower-bounded by $\overline{D}_T+D_0$. This suggests that the uncertainty characterized by $\overline{D}_T$ for the robust controller design depends on the exploration phase through~\eqref{eq:opt_problem_Ue}. This couples the exploration phase and robust control, consequently resulting in a \emph{dual effect} of the proposed controller.

Upon application of the exploration input determined by solving \eqref{eq:opt_problem}, new data is obtained. Combining the new data with the prior information, the MAP estimate~\eqref{eq:LMS} is used to obtain obtain improved/updated estimates $\hat{A}_T$, $\hat{B}_T$ and a new bound $D_T^{-1}$ on the uncertainty. Parameters $\hat{\theta}_T$ are projected on $\mathbf{\Theta}_0$, resulting in $\tilde{\theta}_T$ and corresponding to matrices $\tilde{A}_T$, $\tilde{B}_T$. This ensures that the uncertainty over the parameters \eqref{eq:projection} does not increase. Subsequently, after $T$ time steps, the designed gain-scheduled controller with the new scheduling variable $\Delta_s=\begin{bmatrix}\tilde{A}_T-\hat{A}_0&\tilde{B}_T-\hat{B}_0\end{bmatrix}$ can be applied. Using \eqref{closedloop}, this controller can be explicitly written as a state feedback control law $K$:\nolinebreak
\begin{align}
\label{eq:K_new}
u_k&=K_\mathrm{x} x_k+K_\mathrm{s} w_k^\mathrm{s}\\
&=K_\mathrm{x} x_k+K_\mathrm{s}\left(\left(\tilde{A}_T-\hat{A}_0\right)x_k+\left(\tilde{B}_T-\hat{B}_0\right)u_k\right)\nonumber\\
&=\left(I_{n_\mathrm{u}}-K_\mathrm{s}\left(\tilde{B}_T-\hat{B}_0\right)\right)^{-1}\left(K_\mathrm{x}+K_\mathrm{s}\left(\tilde{A}_T-\hat{A}_0\right)\right)x_k\nonumber\\\nonumber
&=:K x_k.
\end{align}\nolinebreak
Note that $(I-K_\mathrm{s}(\tilde{B}_T-\hat{B}_0))$ is non-singular due to the equivalence in \cite[Thm.~2]{scherer2001lpv}\footnote{The feasibility of the LMIs \eqref{eq:LMIs_GS} implies the well-posedness of the closed loop \eqref{closedloop}, which in turn results in the non-singularity of $(I-K_\mathrm{s}(\tilde{B}_T-\hat{B}_0))$ for all $\Delta_\mathrm{s} \in \bm{\Delta}=\{\Delta: \Delta^\top \Delta \preceq D_0^{-1}\}$ (cf. Lemma \ref{lem:projection-bounds}, Part II).}. The overall procedure is summarized in Algorithm~\ref{alg:main}. 

\begin{algorithm}[H]
\caption{Dual control using gain-scheduling}
\label{alg:main}
\begin{algorithmic}[1]
\State Use a Gaussian prior (cf. Assumption \ref{a1}) to determine initial estimates $\hat{A}_0,\,\hat{B}_0$ and uncertainty bound $D_0^{-1}$.
\State Specify confidence level $\delta\in(0,1)$, $H_2$ performance index $\gamma_\mathrm{p}$~\eqref{eq:h2perf}, exploration length $T$ and $L$ frequencies.
\State Compute scalar $l_1$ \eqref{eq:l1}. Compute bounds $\Gamma_\mathrm{v}$ and $\gamma_\mathrm{y}$ in \eqref{eq:tf_prop}, and  via methods described in Appendix \ref{appendix:gamma} and \ref{appendix:gamma_y}.
\Statex \textit{Alternative:} Set $\beta \ll 1$, e.g., $\beta=10^{-10}$, and compute bound $\Gamma_\mathrm{v}$ and $\gamma_\mathrm{y}$ \eqref{eq:tf_prop} via the scenario approach as described in Appendix \ref{appendix:sample}.
\State Select initial candidate $\tilde{U}$ \eqref{eq:L_choice}.
\State Solve the optimization problem~\eqref{eq:opt_problem} for different values $\epsilon,\,\lambda_\mathrm{s},\,\lambda_\mathrm{u}>0$ (e.g., via line-search in an outer loop).
\Statex $\Rightarrow$ Sinusoidal exploration sequence amplitudes $a_i$ corresponding to their frequencies $\omega_i$ (cf. \eqref{eq:exploration_controller}).
\Statex $\Rightarrow$ Gain-scheduled controller parameters $K_\mathrm{s}$, $K_\mathrm{x}=MN^{-1}$ (cf. \eqref{eq:K_new}).
\State Apply the exploration input as in \eqref{eq:exploration_controller} for $k=0,\dots,T-1$.
\State Update estimates $\hat{A}_T,\,\hat{B}_T$ using new data (c.f. MAP estimation \eqref{eq:LMS}) and obtain projected parameters $\tilde{A}_T,\,\tilde{B}_T$ \eqref{eq:projection}.
\State Compute the state-feedback $K$ dependent on $\tilde{A}_T,\,\tilde{B}_T$ \eqref{eq:K_new}.
\State Apply the feedback $u_k=K x_k$, $k \geq T$.
\end{algorithmic}
\end{algorithm}

\subsection{Theoretical analysis}
\label{sec:dual_2}
The following result proves that Algorithm~\ref{alg:main} leads to a controller with closed-loop guarantees.
\begin{theorem}
\label{thm:main}
 Let Assumption \ref{a1} hold.  Suppose Problem \eqref{eq:opt_problem} is feasible, and Algorithm \ref{alg:main} is applied. Then, with probability\footnote{In case the constants $\Gamma_\mathrm{v}$ and $\gamma_\mathrm{y}$ are both computed using the scenario approach (cf. Appendix \ref{appendix:sample}), then each constant is only valid with confidence of $1-\beta$ and correspondingly the robust performance bound only holds with probability $1-3\delta-2\beta$.} $1-3\delta$, the closed-loop system \eqref{closedloop} with $u_k=Kx_k$ \eqref{eq:K_new} satisfies the $H_2$ performance bound $\gamma_\mathrm{p}$ \eqref{eq:h2perf}.
\end{theorem}
\begin{proof}
Firstly, we recall that Lemma \ref{lemma:robust} guarantees the performance bound \eqref{eq:h2perf}, assuming suitable bounds on $\Delta_\mathrm{s}$ and $\Delta_\mathrm{u}$. Then, we show that the exploration inequalities in combination with parameter projection in \eqref{eq:projection} ensure the bounds on  $\Delta_\mathrm{s}$ and $\Delta_\mathrm{u}$.

\textbf{Part I.} Recall that $R_\mathrm{s}^{-1}=D_0$ and $R_\mathrm{u}^{-1}=\overline{D}_{\mathrm{post}}$. According to Lemma \ref{lemma:robust}, the satisfaction of the robust control LMIs in  \eqref{eq:LMIs_GS} guarantees that the robust gain-scheduled controller $u_k~=~MN^{-1}x_k + K_\mathrm{s}w_k^\mathrm{s}$ ensures the $H_2$ performance bound \eqref{eq:h2perf}, if $\Delta_\mathrm{s}^\top \Delta_\mathrm{s} \preceq D_0^{-1}$, $\Delta_\mathrm{u}^\top \Delta_\mathrm{u} \preceq \overline{D}_\mathrm{post}^{-1}$. Thus, only the bounds $\Delta_\mathrm{s}^\top \Delta_\mathrm{s} \preceq D_0^{-1}$, $\Delta_\mathrm{u}^\top \Delta_\mathrm{u} \preceq \overline{D}_\mathrm{post}^{-1}$ remain to be shown.

\textbf{Part II.} Lemma \ref{lem:projection-bounds} shows that these bounds hold jointly with probability $1-3\delta$. \eqref{eq:finalbound}

Hence, the closed-loop system \eqref{closedloop} with $u_k=Kx_k$ \eqref{eq:K_new} satisfies the performance bound \eqref{eq:h2perf} with probability $1-3\delta$.\qedhere
\end{proof}

This result summarizes the theoretical properties of the proposed dual control approach, ensuring a desired performance specification \eqref{eq:h2perf} using a targeted exploration strategy.

\subsection{Discussion}

In what follows, we discuss the main features of the proposed approach as well as connections to existing works.

\textit{Summary - proposed approach:} The proposed approach, as outlined in Algorithm~\ref{alg:main}, yields a  sequential learning and control strategy jointly from SDP \eqref{eq:opt_problem}: $\mathrm{(i)}$ a harmonic exploration input to generate data for maximum a posteriori estimation; $\mathrm{(ii)}$ a linear state-feedback (exploitation) which is implemented after exploration. 
The resulting  state-feedback explicitly depends on the generated data in terms of the mean estimate and (with high probability) guarantees stability and performance (cf. Theorem~\ref{thm:main}). In general, direct state measurement, which is required for state-feedback, is a practical limitation. However, the assumption of i.i.d. process noise that is normally distributed with zero mean and known variance, and direct state measurement, allow the usage of standard MAP estimation which simplifies the exposition of the main results.
The lower-bound on excitation that is derived is comparable to the Fisher information matrix \cite{pronzato2008optimal}, with the key difference that in the design process that follows, uncertainty arising from prior information is robustly accounted for.
%
The proposed control strategy has a dual effect in the sense that the choice of the harmonic exploration input shapes the uncertainty bound of the MAP estimate after exploration (cf. Lemma~\ref{lem1} and Theorem~\ref{thm1}), which enters in the robust performance bound during exploitation (cf. Lemma~\ref{lemma:robust}). 
A main technical tool for establishing this link is the theory of spectral lines applied to the harmonic exploration (cf. Section~\ref{sec:exploration}).
As a result, the proposed exploration is \emph{targeted} in the sense that the amplitudes at different frequencies influence both the magnitude and the shape/direction of the uncertainty after exploration, and hence the uncertainty that is accounted for in the (robust) control design. 
%
%
Our dual control approach uses Problem~\eqref{eq:opt_problem}, which minimizes the energy of the exploration input such that a specified performance can be guaranteed for the closed loop after the exploration phase. With trivial modifications, it is also possible to solve the converse problem, i.e., optimizing the closed-loop performance after exploration subject to a constraint on the exploration energy.

\textit{Limitations:} Performance synthesis through robust gain-scheduling based on a common Lyapunov function provides a simple design, however, introduces conservatism. This conservatism could be reduced by using a parameter-dependent Lyapunov function~\cite{dinh2005parameter}.
The parametrized gain-scheduling feedback computer prior to exploration can also be improved by re-optimization using the updated estimates $\hat{A}_T, \hat{B}_T$ and bound $D_T^{-1}$. The presented theory ensures that (with high probability) this controller achieves at least the desired $H_2$ performance bound $\gamma_\mathrm{p}$.
\subsubsection*{Related approaches}
The proposed dual control algorithm is similar to, and inspired by, recent robust dual control approaches \cite{umenberger2019robust, ferizbegovic2019learning}, albeit with two crucial differences that are discussed as follows.

Firstly, in the exploration phase, we consider \emph{harmonic} inputs of the form $u_k=\sum_{i=1}^{L}a_i \cos(2 \pi \omega_i k)$, whereas \cite{umenberger2019robust, ferizbegovic2019learning, venkatasubramanian2020robust} suggest inputs $u_k=Kx_k+v_k$ with a matrix $K$ and a zero-mean Gaussian random variable $v_k$.
The frequencies $\omega_i$ in the harmonic input signal are specified a priori, and allow for an intuitive tuning based on possible prior knowledge about the system.
Subsequently, the proposed dual control approach determines the corresponding amplitudes $a_i$, weighting the frequency components to achieve a targeted uncertainty reduction. 
Most importantly, the influence of the exploration input on the uncertainty bound after exploration can be quantified \emph{a priori}, in a rigorous manner as derived in Section \ref{sec:prelim_spectral}, using the theory of spectral lines \cite{sarker2020parameter}.
On the contrary, the bounds obtained by~\cite{umenberger2019robust, ferizbegovic2019learning, venkatasubramanian2020robust} do \textit{not} yield any guarantees for the actual exploration since the empirical covariance is approximated in a heuristic way via the worst-case covariance.

Secondly, in the exploitation phase, we consider a \emph{gain-scheduling controller}, where the difference of the parameter estimates before and after exploration plays the role of the scheduling parameter. This enables us to provide a priori performance and stability guarantees. On the contrary, \cite{umenberger2019robust} and \cite{ferizbegovic2019learning} consider the simplification that the estimate after exploration is equal to initial estimate.  Given that the initial estimate is uncertain, \cite{umenberger2019robust} and \cite{ferizbegovic2019learning} require a re-design step after exploration. Furthermore, such an additional online design step, which could be equivalently added in our proposed approach to further improve performance, implies that the exploration is not directly targeted to the final robust design problem, and no \emph{a priori} guarantees can be given. We note that a similar gain-scheduling approach to robust dual control was also suggested in our preliminary work \cite{venkatasubramanian2020robust}, however, employing a random exploration strategy without excitation guarantees and satisfying a quadratic performance bound. Furthermore, since \cite{venkatasubramanian2020robust} considers a random exploration input (as in \cite{umenberger2019robust} and \cite{ferizbegovic2019learning}), the resulting guarantees only hold under restrictive assumptions (cf.~\cite[Assumption~2]{venkatasubramanian2020robust}).

Finally, we note that the considered control problem could also be addressed using results from regret analysis and sample complexity bounds of LQR design based on system identification. For example, in~\cite{dean2017sample}, a random excitation is followed by the least-squares estimate and a robust LQR design. This also provides performance/regret bounds with respect to the optimal policy in dependence of the number of roll-outs, i.e., numerous trials with re-start. In contrast to the proposed approach, the exploration is not targeted, and in general, the guarantees are qualitative in nature, i.e., in the limit of multiple roll-outs, the performance approaches the optimum, but, sharp performance specifications for the first roll-outs cannot be given.

\section{Numerical Example}
In this section, we demonstrate the practical applicability of the proposed sequential robust dual control approach through a numerical example, which is `hard to learn'. Numerical simulations were performed on MATLAB using CVX \cite{cvx} in conjunction with the default SDP solver SDPT3. First, we illustrate the benefits of the proposed targeted exploration strategy \eqref{eq:exp_problem} compared to a common random exploration strategy (Section \ref{sec:exploration}). Subsequently, we implement the overall dual control strategy (Algorithm \ref{alg:main}) and demonstrate the trade-off between exploration energy and performance achieved after exploration (Section \ref{sec:dual_control}).

\subsection{Problem setup}
We consider a linear system \eqref{sys} with $\sigma_\mathrm{w}^2=1$ and
\begin{align}\label{eq:exsys}
A_\mathrm{tr}=\begin{bmatrix}
0.49 & 0.49 & 0 & 0\\ 0 & 0.49 & 0.49 & 0\\ 0 & 0 & 0.49 & 0.49 \\ 0 & 0 & 0 & 0.49
\end{bmatrix},\;
B_\mathrm{tr}=\begin{bmatrix}
0 \\ 0 \\ 0\\  0.49
\end{bmatrix}
\end{align}
which belongs to a class of linear systems identified as `hard to learn' \cite{tsiamis2021linear}. The generalized error \eqref{gz} is characterized by the known matrix $C=I$. We generate one random noise sequence $w_k \sim \mathcal{N}(0,1)$ for the exploration time horizon $T=100$, which is used for all the following simulations. We consider a Gaussian prior (cf. Assumption \ref{a1}) with initial estimate $\hat{\theta}_\mathrm{prior}$ and an initial uncertainty bound $D_0^{-1}~=~5\cdot~10^{-3}I$.
The estimate $\hat{\theta}_\mathrm{prior}$ is sampled randomly such that $(\hat{\theta}_\mathrm{prior}~-~\theta_\mathrm{tr})^\top~(D_0~\otimes~I_{n_\mathrm{x}})~ (\hat{\theta}_\mathrm{prior}-\theta_\mathrm{tr}) \leq 1$. We set $\epsilon=0.5$ and select the probability of violation $\delta=0.01$. The constants $\Gamma_\mathrm{v}$, $l$ are computed using the scenario approach (cf. Appendix \ref{appendix:sample}) with confidence level $\beta=10^{-10}$. 

\subsection{Targeted exploration}
In this section, we demonstrate the targeted exploration strategy as described in Section \ref{sec:exploration}. In addition, we compare the targeted exploration strategy with a random exploration strategy that applies normally distributed inputs, i.e., the inputs $u_k \sim \mathcal{N}(0,1)$. For a fair comparison, the random exploration inputs are scaled such that both strategies apply inputs with the same energy for exploration. We then evaluate and compare the level of excitation achieved.

In the context of targeted exploration, the objective is to ensure a given lower bound on excitation $D_T \succeq \overline{D}_T$, and targeted exploration inputs are obtained by solving the exploration problem \eqref{eq:exp_problem}. In our simulations, we set $T~=~100$ and select $L=10$ frequencies from the set $\Omega_{100}$ to yield $\omega_i=\{0,0.1,0.2,0.3,0.4,0.5,0.6,0.7,0.8,0.9\}$, $i=1,...,L$. The solution of \eqref{eq:exp_problem} provides the resulting amplitudes $a_i,\,i=1,...,L$ corresponding to the sinusoids with frequencies $\omega_i,\,i=1,...,L$. Given $T$, $a_i$ and $\omega_i$, the input signal \eqref{eq:exploration_controller} and its energy $\sum_{k=1}^T u_k^2$ can be computed.

In our simulations, we implement the proposed targeted exploration and random exploration for different initial uncertainty bounds $D_0^{-1}$. Due to the chained structure of $A_\mathrm{tr}$ with weakly-coupled states, it is difficult to excite the first state. Hence, we compare the excitation achieved by targeted and random exploration corresponding to the first state, i.e., the $(1,1)$-element of $D_T$ denoted by $D_{T_{11}}$. We impose a constraint on $D_{T_{11}}$ by setting $\overline{D}_{T_{11}} \geq 10^6$. For a trial, we multiply the initial excitation bound $D_0$ by a factor $\alpha_i$ to achieve $D_{0_i}=\alpha_{i} D_0$. We randomly sample a prior estimate $\hat{\theta}_{\mathrm{prior},i}$ such that $(\hat{\theta}_{\mathrm{prior},i}-\theta_\mathrm{tr})^\top (D_{0_i} \otimes I_{n_\mathrm{x}}) (\hat{\theta}_{\mathrm{prior},i}-\theta_\mathrm{tr}) \leq 1$, and compute the corresponding constants $\Gamma_{\mathrm{v},i}$ and $l_i$ via the scenario approach (cf. Appendix \ref{appendix:sample}). We run 10 trials for each of the following values $\alpha_i \in \{10^0,10^1,10^2,10^3,10^4\}$ wherein each trial uses a different randomly sampled initial estimate $\hat{\theta}_\mathrm{prior}$ and a different sequence of random inputs for random exploration. Each trial comprises: $\mathrm{(i)}$ determining the initial excitation bound $D_{0_i}$, prior estimate $\hat{\theta}_{\mathrm{prior},i}$, and the corresponding constants $\Gamma_{\mathrm{v},i}$ and $l_i$, $\mathrm{(ii)}$ solving the exploration problem \eqref{eq:exp_problem} to obtain targeted exploration inputs, and $\mathrm{(iii)}$ generating random exploration inputs with the same energy as the targeted inputs. We iterate Problem \eqref{eq:exp_problem} over $\tilde{U}$ \eqref{eq:L_choice}, as described in Section \ref{subsec:exp_lowerbound} to reduce the suboptimality of convex relaxation. Given the targeted and random exploration inputs, we get observed data sets from 10 trials corresponding to each $\alpha_i$. From the datasets, we compute the mean value of finite excitations achieved by targeted exploration $D_{T_{T,i,11}}$ and random exploration $D_{T_{R,i,11}}$. From Fig. \ref{fig:goaldt}, it can be observed that the targeted exploration inputs achieve the desired excitation for all tested initial uncertainty levels, i.e., $D_{T_{T,i,11}} \geq 10^6$ for all trials. The achieved excitation is higher than the desired excitation due to inherent conservatism of the targeted exploration strategy, which utilizes worst-case bounds $\Gamma_{\mathrm{v}}$ and $l$. The strategy is more conservative for large initial uncertainty (small $\alpha$). On the other hand, random exploration achieves significantly smaller excitation with the same energy by a factor of approximately $10$. Furthermore, random exploration has a large variance in the achieved exploration due to the stochasticity introduced by the random inputs. Hence, targeting the exploration with respect to the desired excitation can have enormous advantages over random exploration.
\begin{figure}[h!]
\begin{center}
\includegraphics[trim={0.2cm 0.1cm 1cm 0.5cm}, clip, width=0.48\textwidth]{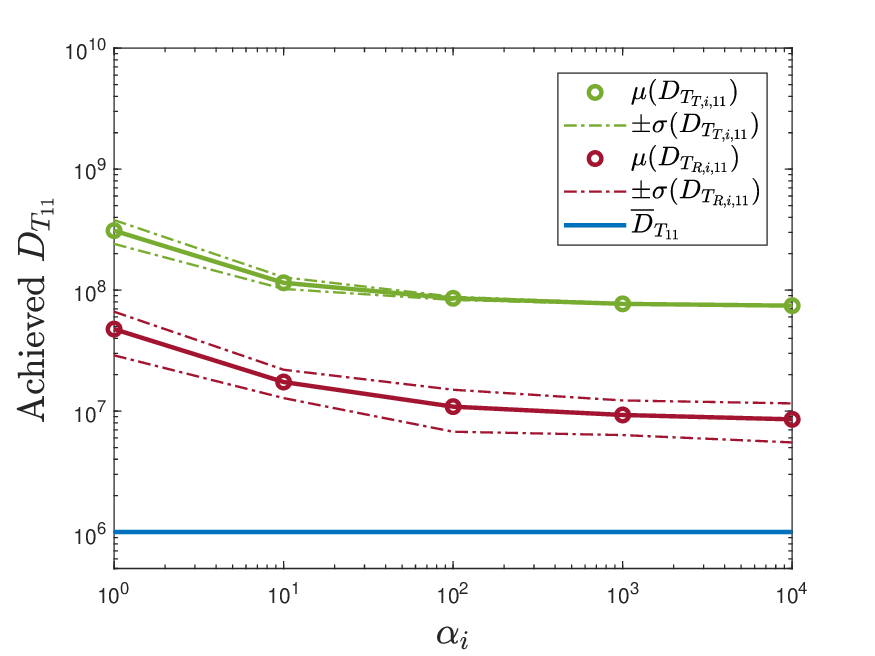}
\end{center}
\caption{Illustration of $\mathrm{(i)}$ the mean and standard deviation of the excitation achieved by targeted exploration $D_{T_{T,11}}$, and by random exploration $D_{T_{R,11}}$, for different values of $\alpha_i$ that scales the initial uncertainty bound $D_0$, and $\mathrm{(ii)}$ the desired excitation $\overline{D}_{T_{11}}=10^6$.}
\label{fig:goaldt} 
\end{figure}

\subsection{Dual control with robust gain-scheduling}
In this section, we use the proposed dual control strategy to study the trade-off between the desired $H_2$ performance and the resulting exploration energy $\gamma_\mathrm{e}$ \eqref{eq:min_energy_cost}. We consider again an initial uncertainty bound, i.e., $D_0^{-1} = 5\cdot 10^{-3}I$ and a fixed initial estimate $\hat{\theta}_{\mathrm{prior}}$ satisfying $(\hat{\theta}_\mathrm{prior}~-~\theta_\mathrm{tr})^\top (D_0~\otimes~I_{n_\mathrm{x}}) (\hat{\theta}_\mathrm{prior}-\theta_\mathrm{tr})\leq 1$. For $H_2$ performance, we consider a $\gamma_p>0$ (Def. \ref{def:h2perf}). In order to determine the efficacy of the proposed dual controller, we additionally determine the performance of a nominal $H_2$ controller with exact model knowledge, which is denoted by $\underline{\gamma}_\mathrm{p}=2.65$, and the performance of a robust $H_2$ controller computed based on the initial uncertainty bound $D_0^{-1}$, which is denoted by $\overline{\gamma}_\mathrm{p}=3.05$. We solve Problem \eqref{eq:opt_problem} for different values of $\gamma_\mathrm{p}$ in the interval $[\underline{\gamma}_\mathrm{p},\overline{\gamma}_\mathrm{p}]$. From Fig. \ref{fig:gegp}, it can be observed as we want to achieve better performance (decrease $\gamma_\mathrm{p}$), the required exploration energy $\gamma_\mathrm{e}$ increases. There exists a large range of $\gamma_\mathrm{p}$ where we can incrementally achieve better performance by increasing exploration energy. However, there exists a minimal performance $\gamma_\mathrm{p} \approx 2.85$ which we can guarantee a priori with the proposed approach, even with arbitrarily large exploration energy. The fact that this is larger than the ideal performance $\underline{\gamma}_\mathrm{p}=2.65$ is partially due to the conservatism of the robust gain-scheduling synthesis based on a common Lyapunov function.

\begin{figure}[h!]
\begin{center}
\includegraphics[trim={0.1cm 0cm 1cm 0.5cm}, clip, width=0.48\textwidth]{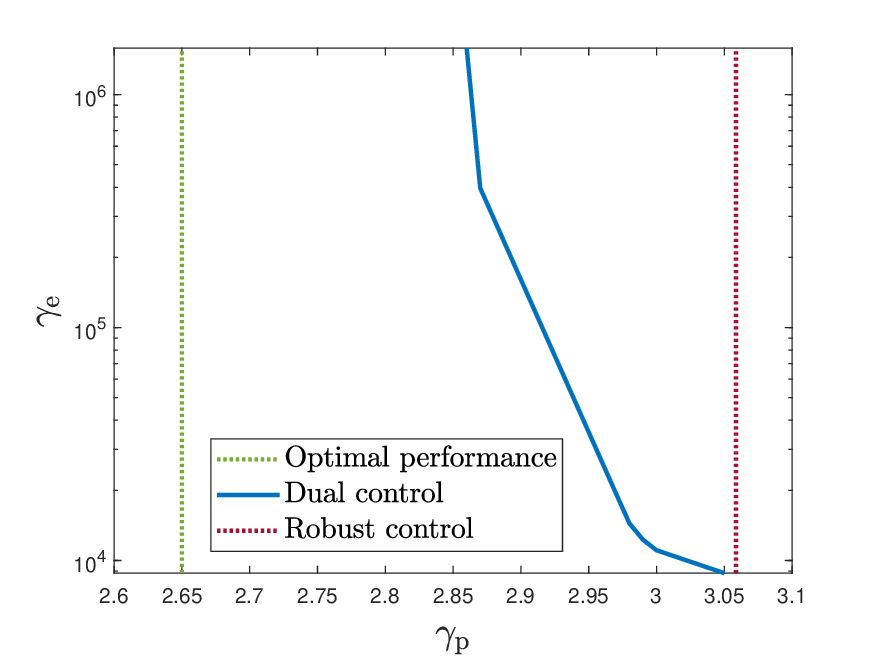}
\end{center}
\caption{Illustration of the exploration energy $\gamma_\mathrm{e}$ required to achieve a certain $H_2$ performance $\gamma_\mathrm{p}$. For comparison, the $H_2$ performance of a robust controller based on the initial uncertainty bound $D_0^{-1}$, and the optimal $H_2$ performance based on true system knowledge are provided.}
\label{fig:gegp} 
\end{figure}

For fixed values of $\epsilon$, $\lambda_\mathrm{s}$, and $\lambda_\mathrm{u}$, the average execution time of Problem \eqref{eq:opt_problem} over 10 trials, with and without iterations over $\tilde{U}$ is 2.74 seconds and 0.85 seconds, respectively. In the considered numerical example, satisfactory performance was observed for fixed $\epsilon=0.5$, and line-search over $\lambda_\mathrm{s}$ and $\lambda_\mathrm{u}$. In the numerical example, the line-search involved 10 iterations. The average execution time for solving Problem \eqref{eq:opt_problem} in a loop with line-search, over 10 trials, is approximately 28 seconds. In general, Problem \eqref{eq:opt_problem} is a standard SDP, and the computational effort scales polynomially with the dimension of the state $n_\mathrm{x}$ and the control input $n_\mathrm{u}$. The simulations were carried out on a system with an Intel(R) Core(TM) i7-9750H CPU @ 2.60GHz Processor and 16.0 GB RAM.

Overall, the simulation results corroborate the benefits of the proposed dual-control strategy. Given an initial prior and uncertainty bound, the targeted exploration strategy guarantees an \textit{a priori} lower bound on excitation, which in turn influences the uncertainty bound over model parameters after exploration. Furthermore, dual control with gain-scheduling demonstrates the \textit{targeted} nature of the exploration inputs, i.e., exploration is only carried out in order to improve the accuracy of system parameters and increase performance beyond the performance of a robust controller with knowledge of initial uncertainty.

\section{Conclusion}
In this article, we have presented a sequential dual control strategy that involves a targeted exploration phase with harmonic inputs that can guarantee \textit{a priori} excitation bounds, and a robust control phase based on gain scheduling that can guarantee $H_2$ performance after exploration. To simplify the exposition, we have considered Schur stable systems with process noise and directly measurable states. The main technical tool used to design the exploration phase is the theory of spectral lines. The amplitudes of the harmonic inputs are optimized so as to shape the uncertainty bound of the parameter estimates after exploration, which are accounted for by the controller based on gain scheduling, and thereby encapsulate the dual effect. We have demonstrated the applicability of the proposed dual control strategy, and clear benefits of targeted exploration, with numerical application to a system from a class that is `hard to learn'.

In summary, to the best knowledge of the authors, the presented approach provides the first solution of the \emph{robust dual control problem}, which is $\mathrm{(i)}$ computationally tractable, and $\mathrm{(ii)}$ links a targeted exploration with an exploitation phase yielding robust guarantees without resorting to approximations or heuristics to mimic a dual effect.

\appendices
\section{Bound on $\tilde{V}\tilde{V}^\mathsf{H}$}\label{appendix:gamma}
In order to determine a bound of the form $\tilde{V}\tilde{V}^\mathsf{H} \preceq \Gamma_\mathrm{v}$, we first determine a bound on $\|\tilde{V}\|$ based on the sub-matrices of $\tilde{V}$.  The matrix $\tilde{V}=V_\mathrm{tr}~-~\hat{V}~\in~ \mathbb{C}^{n_\mathrm{\phi} \times Ln_\mathrm{u}}$ can be denoted in terms of its sub-matrices as $\tilde{V}~=~[\tilde{V}_1,\cdots,\tilde{V}_{L}]~=~[V_1
~-~\hat{V}_1,\cdots,V_{L}~-~\hat{V}_{L}]$. 

\begin{proposition}\label{prop:sub_matrix_norm}
Suppose $\|\tilde{V}_i\| \leq \overline{\gamma}_\mathrm{v},\,\forall i=1,\dots,L$, then $\|\tilde{V}\|\leq \overline{\gamma}_\mathrm{v} \sqrt{L}$.
\end{proposition}

\begin{proof}
For $x=\left[ x_1^\top,\dots ,\, x_{L}^\top\right]\in \mathbb{C}^{Ln_\mathrm{u}}$, we have
\begin{align*}
\|\tilde{V}\|=\sup_{\|x\|\leq 1} \|\tilde{V}x\| &=\sup_{\|x\|\leq 1} \|\tilde{V}_1 x_1+\cdots+\tilde{V}_{L} x_{L}\|.
\end{align*}
By using the triangle inequality, we have
\begin{align*}
\sup_{\|x\|\leq 1} \|\tilde{V}_1 x_1+\cdots+\tilde{V}_{L} x_{L}\| \leq \sup_{\|x\| \leq 1} \sum_{i=1}^{L}  \|\tilde{V}_i x_i\|.
\end{align*}
Finally, by using the sub-multiplicativity property, and subsequently, the Cauchy-Schwarz inequality, we have
\begin{align*}
\sup_{\|x\| \leq 1} \sum_{i=1}^{L}  \|\tilde{V}_i x_i\| \leq  \overline{\gamma}_\mathrm{v} \sup_{\|x\| \leq 1} \sum_{i=1}^{L} \|x_i\| = \overline{\gamma}_\mathrm{v} \sqrt{L}. 
\end{align*}\qedhere
\end{proof}

The matrices $\tilde{V}_i$ are components of the transfer matrix of the following open loop system evaluated at frequencies $\omega_i \in \Omega_T,\, i=1,\dots,L$:
\begin{equation}\label{eq:open_loop_new}
\begin{split}
\xi_{k+1}&=\begin{bmatrix}
\hat{A}_0 & 0\\0 & A_\mathrm{tr}
\end{bmatrix} \xi_k + \begin{bmatrix}
\hat{B}_0 \\ B_\mathrm{tr}
\end{bmatrix} u_k \\
z_k & = \begin{bmatrix}
I & -I\\0 & 0
\end{bmatrix}\xi_k + \begin{bmatrix}
0\\I
\end{bmatrix}u_k
\end{split}
\end{equation}
where $\xi_k=\begin{bmatrix}\hat{x}_k\\x_k\end{bmatrix}$. The prior uncertainty block is denoted as
\begin{align}
\Delta_0=\begin{bmatrix}A_\mathrm{tr}-\hat{A}_0 & B_\mathrm{tr}-\hat{B}_0
\end{bmatrix}.
\end{align}
Given $\Delta_0^\top \Delta_0 \preceq D_0^{-1}$ (cf. Assumption \ref{a1} and Lemma \ref{lem1}), we can compute $\overline{\gamma}_\mathrm{v}$ such that $\|\tilde{V}_i\| \leq \overline{\gamma}_\mathrm{v}$ using the (robust) $H_\infty$-norm, which is equal to the $\ell_2$-gain of the system \eqref{eq:open_loop_new}. The $\ell_2$-gain for the channel $u \rightarrow z$ is guaranteed to be robustly smaller than $\overline{\gamma}_\mathrm{v}$ if, for $P_\mathrm{p}=\begin{pmatrix}
 -\overline{\gamma}_\mathrm{v} I & 0\\0 & \frac{1}{\overline{\gamma}_\mathrm{v}}I
\end{pmatrix}$, the following inequality holds for some $\epsilon>0$ (cf. \cite{scherer2000linear}):
\begin{align}
\sum_{k=0}^{\infty}\begin{pmatrix}u_k\\z_k\end{pmatrix}^\top P_\mathrm{p} \begin{pmatrix}u_k\\z_k\end{pmatrix}\leq -\epsilon \sum_{k=0}^{\infty}u_k^\top u_k.
\end{align}

 Accounting for the prior uncertainty bound and re-writing \eqref{eq:open_loop_new} in the standard form, we get
\begin{equation}\label{eq:standard_openloop}
\begin{gathered}
\begin{bmatrix}
\xi_{k+1}\\z_k^u\\z_k
\end{bmatrix}
=
\begin{bmatrix}
\begin{bmatrix}\hat{A}_0 & 0\\ 0 & \hat{A}_0 \end{bmatrix} & \begin{bmatrix}0 \\ I \end{bmatrix} & \begin{bmatrix} \hat{B}_0 \\ \hat{B}_0 \end{bmatrix} \\
\begin{bmatrix}
0 & I\\ 0 & 0
\end{bmatrix}& 0 & \begin{bmatrix}
0 \\ I
\end{bmatrix}  \\
\begin{bmatrix}
I & -I\\0 & 0
\end{bmatrix} & 0 & \begin{bmatrix} 0\\1\end{bmatrix}
\end{bmatrix} \begin{bmatrix}
\xi_{k}\\w_k^u\\u_k
\end{bmatrix},\\
w_k^\mathrm{u} = \Delta_0 z_k^\mathrm{u}.
\end{gathered}
\end{equation}
Given \eqref{eq:standard_openloop}, $\overline{\gamma}_\mathrm{v}$ is a valid $H_\infty$-norm bound if there exists a matrix $N$ and a scalar $\lambda_\mathrm{v}>0$ such that
\begin{scriptsize}
\begin{align}\label{eq:openloop_Hinf_Schur}
\begin{bmatrix}
\begin{array}{c|c}
\begin{bmatrix}
-N & 0 & 0\\0 & -\lambda_\mathrm{v} I & 0\\0 & 0 & -\overline{\gamma}_\mathrm{v} I
\end{bmatrix} & \star \\ \hline
\begin{bmatrix}
\begin{bmatrix}\hat{A}_0  & 0\\ 0 & \hat{A}_0  \end{bmatrix}N & \begin{bmatrix}0 \\ I \end{bmatrix} & \begin{bmatrix} \hat{B}_0 \\ \hat{B}_0 \end{bmatrix} \\
\begin{bmatrix}0 & I \\ 0 & 0\end{bmatrix}N & 0 & \begin{bmatrix}0 \\ I \end{bmatrix}  \\
\begin{bmatrix}
I & -I\\0 & 0
\end{bmatrix}N & 0 & \begin{bmatrix} 0\\I\end{bmatrix}
\end{bmatrix} & \begin{bmatrix}
-N & 0 & 0\\0 & -\frac{1}{\lambda_\mathrm{v}}D_0 & 0\\0 & 0 & -\overline{\gamma}_\mathrm{v} I
\end{bmatrix}
\end{array}
\end{bmatrix} \prec 0.
\nonumber\\\text{ }
\end{align}
\end{scriptsize}
By defining $X=N^{-1}$ and multiplying the Schur complement of  \eqref{eq:openloop_Hinf_Schur} from left and right by $\text{diag}(N^{-1},I,I)$, we get:
\begin{align}\label{eq:openloop_Hinf_v}
\begin{split}
\left[\begin{array}{c}
*\\ *\\ \hline
*\\ *\\ \hline
*\\ *
\end{array}\right]^\top & \left[ \begin{array}{c|c|c}
\begin{matrix}-X & 0 \\ 0 & X \end{matrix} &
\begin{matrix}0&0\\0&0\end{matrix}&
\begin{matrix}0&0\\0&0\end{matrix}\\ \hline
\begin{matrix}0&0\\0&0\end{matrix}& \lambda_\mathrm{v} P_\mathrm{u} & \begin{matrix}0&0\\0&0\end{matrix}\\ \hline
\begin{matrix}0&0\\0&0\end{matrix}&
\begin{matrix}0&0\\0&0\end{matrix}&
P_\mathrm{p}
\end{array}\right]\\ & \times \left[\begin{array}{ccc}
I & 0 & 0\\
\begin{bmatrix}\hat{A}_0 & 0\\ 0 & \hat{A}_0 \end{bmatrix} & \begin{bmatrix}0 \\ 1 \end{bmatrix} & \begin{bmatrix} \hat{B}_0 \\ \hat{B}_0 \end{bmatrix}\\ \hline
0 & I & 0\\
\begin{bmatrix}
0 & I \\ 0 & 0
\end{bmatrix}& 0 & \begin{bmatrix}
0 \\ I
\end{bmatrix} \\ \hline
0 & 0 & I\\
\begin{bmatrix}
I & -I\\0 & 0
\end{bmatrix} & 0 & \begin{bmatrix} 0\\I\end{bmatrix}
\end{array} \right] \prec 0,
\end{split}
\end{align}
where $P_\mathrm{u}=\begin{bmatrix}-I&0\\0&D_0^{-1}\end{bmatrix}$ and $P_\mathrm{p}$ is of the form given earlier. Using \cite{scherer2001lpv}, $\overline{\gamma}_\mathrm{v}$ is a valid $H_\infty$-norm bound if there exists a positive definite matrix $X=X^\top \succ 0$ satisfying \eqref{eq:openloop_Hinf_v}. $N \succ 0$ follows from \eqref{eq:openloop_Hinf_v} with $N=X^{-1} \succ 0$.

A solution of  \eqref{eq:openloop_Hinf_Schur} gives $\overline{\gamma}_\mathrm{v}$ from which we can compute the bound 
\begin{align}
\|\tilde{V}\| \leq \overline{\gamma}_\mathrm{v}\sqrt{L}.
\end{align}
Hence,
\begin{align}
    \tilde{V} \tilde{V}^\mathsf{H} \preceq \Gamma_\mathrm{v} := \overline{\gamma}_\mathrm{v}^2 L I 
\end{align}
holds with probability at least $1-\delta$ since the bound on the prior uncertainty block holds with probability $1-\delta$, i.e. $\mathbb{P}(\Delta_0^\top \Delta_0 \preceq D_0^{-1}) = 1-\delta$ (cf. Assumption \ref{a1} and Lemma \ref{lem1}).

\section{Bound on $\|Y_\mathrm{tr}\|$}\label{appendix:gamma_y}
A bound on $\|Y_\mathrm{tr}\|$ can be derived in terms of the bounds on its sub-matrices as in Proposition \ref{prop:sub_matrix_norm} (cf. Appendix \ref{appendix:gamma}). Matrices $\|Y_i \|$ are the transfer matrix of the following open loop system evaluated at frequencies $\omega_i \in \Omega_T,\, i=1,\dots,L$:
\begin{align}\label{eq:open_loop_y}
\nonumber
x_{k+1}&=A_\mathrm{tr}x_k+w_k\\
z_k&=x_k.
\end{align}
Hence, the bound $\overline{\gamma}_\mathrm{y}$ is the $H_\infty$ norm of its transfer matrix, which is equal to the $\ell_2$-gain of the system. The $\ell_2$-gain for the channel $w \rightarrow z$ is guaranteed to be robustly smaller than $\overline{\gamma}_\mathrm{y}$ if, for $P_p=\begin{pmatrix}
-\overline{\gamma}_\mathrm{y} I & 0\\0 & \frac{1}{\overline{\gamma}_\mathrm{y}}I
\end{pmatrix}$, the following inequality holds for some $\epsilon>0$ (cf. \cite{scherer2000linear}):
\begin{align}
\sum_{k=0}^{\infty}\begin{pmatrix}w_k\\z_k\end{pmatrix}^\top P_\mathrm{p} \begin{pmatrix}w_k\\z_k\end{pmatrix}\leq -\epsilon \sum_{k=0}^{\infty}w_k^\top w_k.
\end{align}
Since $A_\mathrm{tr}$ is unknown, we first re-write \eqref{eq:open_loop_y} as
\begin{align}\label{eq:open_loop_y_new}
\nonumber
x_{k+1}&=\hat{A}_0x_k+(A_\mathrm{tr}-\hat{A}_0)x_k +w_k\\
z_k&=x_k.
\end{align}
The prior uncertainty block is denoted as
\begin{equation}
\Delta_0=\begin{bmatrix}
A_\mathrm{tr}-\hat{A}_0 & B_\mathrm{tr}-\hat{B}_0
\end{bmatrix}.
\end{equation}
From Assumption \ref{a1} and Lemma \ref{lem1}, we have $\mathbb{P}(\Delta_0^\top \Delta_0 \preceq D_0^{-1}) = 1-\delta$. Accounting for the prior uncertainty bound and re-writing \eqref{eq:open_loop_y_new} in the standard form, we get
\begin{equation}\label{eq:standard_openloop_y}
\begin{gathered}
\begin{bmatrix}
x_{k+1}\\z_k^u\\z_k
\end{bmatrix}
=
\begin{bmatrix}
\hat{A}_0 & I & I \\
\begin{bmatrix}
I\\ 0 
\end{bmatrix}& 0 & 0  \\
I & 0 & 0
\end{bmatrix} \begin{bmatrix}
x_{k}\\w_k^u\\w_k
\end{bmatrix},\\
w_k^\mathrm{u} = \Delta_0 z_k^\mathrm{u}.
\end{gathered}
\end{equation}

Given \eqref{eq:standard_openloop_y}, $\overline{\gamma}_\mathrm{y}$ is a valid $H_\infty$-norm bound if there exists a matrix $N$ and a scalar $\lambda_\mathrm{y}>0$ such that
\begin{footnotesize}
\begin{align}\label{eq:openloop_Hinf_Schur_y}
\begin{bmatrix}
\begin{array}{c|c}
\begin{bmatrix}
-N & 0 & 0\\0 & -\lambda_\mathrm{y} I & 0\\0 & 0 & -\overline{\gamma}_\mathrm{v} I
\end{bmatrix} & \star \\ \hline
\begin{bmatrix}
\hat{A}_0 N & I & I \\
\begin{bmatrix}
N\\ 0 
\end{bmatrix}& 0 & 0  \\
N & 0 & 0
\end{bmatrix} & \begin{bmatrix}
-N & 0 & 0\\0 & -\frac{1}{\lambda_\mathrm{y}} D_0 & 0\\0 & 0 & -\overline{\gamma}_\mathrm{v} I
\end{bmatrix}
\end{array}
\end{bmatrix} \prec 0.
\end{align}
\end{footnotesize}
By defining $X=N^{-1}$ and multiplying the Schur complement of  \eqref{eq:openloop_Hinf_Schur_y} from left and right by $\text{diag}(N^{-1},I,I)$, we get
\begin{align}\label{eq:openloop_Hinf_y}
\left[\begin{array}{c}
*\\ *\\ \hline
*\\ *\\ \hline
*\\ *
\end{array}\right]^\top & \left[ \begin{array}{c|c|c}
\begin{matrix}-X & 0 \\ 0 & X \end{matrix} &
\begin{matrix}0&0\\0&0\end{matrix}&
\begin{matrix}0&0\\0&0\end{matrix}\\ \hline
\begin{matrix}0&0\\0&0\end{matrix}& \lambda_\mathrm{y} P_\mathrm{u} & \begin{matrix}0&0\\0&0\end{matrix}\\ \hline
\begin{matrix}0&0\\0&0\end{matrix}&
\begin{matrix}0&0\\0&0\end{matrix}&
P_\mathrm{p}
\end{array}\right] \left[\begin{array}{ccc}
I & 0 & 0\\
\hat{A}_0 & I & I\\ \hline
0 & I & 0\\
\begin{bmatrix}
I \\ 0
\end{bmatrix}& 0 & 0 \\ \hline
0 & 0 & I\\
I & 0 & 0
\end{array} \right]\prec 0
\end{align}
where $P_\mathrm{u}=\begin{bmatrix}-I&0\\0&D_0^{-1}\end{bmatrix}$ and $P_\mathrm{p}$ is of the form given earlier. Using \cite{scherer2001lpv}, $\overline{\gamma}_\mathrm{y}$ is a valid $H_\infty$-norm bound if there exists a positive definite matrix $X=X^\top \succ 0$ satisfying \eqref{eq:openloop_Hinf_y}. $N\succ 0$ follows from \eqref{eq:openloop_Hinf_y} with $N=X^{-1} \succ 0$.

A solution of  \eqref{eq:openloop_Hinf_Schur_y} gives $\overline{\gamma}_\mathrm{y}$, from which we can compute the bound 
\begin{align}
\|\tilde{Y}\| \leq \gamma_\mathrm{y}:=\overline{\gamma}_\mathrm{y}\sqrt{L}
\end{align}
which holds with probability at least $1-\delta$.

\section{Sample-based constants}\label{appendix:sample}
Bounds $\Gamma_\mathrm{v}$ and $\gamma_\mathrm{y}$ in \eqref{eq:tf_prop} are derived in Appendices \ref{appendix:gamma}-\ref{appendix:gamma_y}. However, these bounds can be conservative, and in what follows, we propose sample-based bounds that can be determined using the `scenario' approach \cite{campi2009scenario}.
\subsubsection*{Scenario approach to estimate $\Gamma_\mathrm{v}$} The bound $\Gamma_\mathrm{v}$ on the uncertain term $\tilde{V} \tilde{V}^\mathsf{H}$ in Appendix \ref{appendix:gamma} is computed under the assumption that $\theta_\mathrm{tr} \in \mathbf{\Theta}_0$. A tighter upper bound $\Gamma_\mathrm{v}$ can be computed by directly estimating an upper bound on $\tilde{V} \tilde{V}^\mathsf{H}$. In order to estimate $\Gamma_\mathrm{v}$, we generate $N_\mathrm{s}$ samples of $\mathrm{vec}(A_i,B_i)=\theta_i \in \mathbf{\Theta}_0$, $i=1,\dots, N_\mathrm{s}$, from the multivariate normal distribution with mean $\hat{\theta}_\mathrm{prior}$ and covariance $\Sigma_{\theta,\mathrm{prior}}$ (cf. Assumption \ref{a1}). Given a probability of violation $\delta$, and the number of uncertain decision variables $d=\frac{n_\phi (n_\phi+1)}{2}$, a lower bound on the number of samples $N_\mathrm{s}$ required to estimate $\Gamma_\mathrm{v}$ with confidence $1-\beta$, is given as \cite{campi2009scenario}:
\begin{align}\label{eq:Ns}
N_\mathrm{s} \geq \frac{2}{\delta}\left(\ln \frac{1}{\beta}+d \right).
\end{align}
For each $\theta_i=\mathrm{vec}(A_i,B_i)$, we evaluate $\tilde{V}_i=\bar{V}_i-\hat{V}$ where the transfer matrix $\bar{V}_i=[\bar{V}_{i,1},\cdots,\bar{V}_{i,L}]$ is computed with $A_i,\,B_i$ (cf. \eqref{eq:tf}). The bound $\Gamma_\mathrm{v}$ can be computed by solving the following SDP:
\begin{align}
\nonumber
\underset{\Gamma_\mathrm{v}}{\min}  &\quad \text{trace}(\Gamma_\mathrm{v}) \\
\text{s.t. }& \quad \Gamma_\mathrm{v} \succeq \tilde{V}_i \tilde{V}_i^\mathsf{H},\;i=1,...,N_\mathrm{s}.
\end{align}

\subsubsection*{Scenario approach to estimate $\gamma_\mathrm{y}$}
Similar to $\Gamma_\mathrm{v}$, the bound $\gamma_\mathrm{y}$ on the uncertain term $\|Y_\mathrm{tr}\|$ in \eqref{eq:tf_prop} is computed under the assumption that $\theta_\mathrm{tr} \in \mathbf{\Theta}_0$ in Appendix \ref{appendix:gamma_y}. A tighter upper bound $\gamma_\mathrm{y}$ on $\|Y_\mathrm{tr}\|$ can be computed by directly estimating an upper bound on $\|Y_\mathrm{tr}\|$.
To this end, given a probability of violation $\delta$, we generate $N_\mathrm{s}$ samples of $\mathrm{vec}(A_i,B_i)=\theta_i~\in~\mathbf{\Theta}_0$, $i=1,\dots, N_\mathrm{s}$, from the multivariate normal distribution with mean $\hat{\theta}_\mathrm{prior}$ and covariance $\Sigma_{\theta,\mathrm{prior}}$ according to \eqref{eq:Ns} with $d=1$. For each $A_i$, we evaluate $\bar{Y}_i=[\bar{Y}_{i,1},\cdots,\bar{Y}_{i,L}]$ (cf. \eqref{eq:tf}). The bound $\gamma_\mathrm{y}$ can be computed as 
\begin{align}
\gamma_\mathrm{y}=\max_{i=1,\dots, N_\mathrm{s} }\| \bar{Y}_i \|.
\end{align}

The constants $\Gamma_\mathrm{v}$ and $\gamma_\mathrm{y}$ computed via the scenario approach ensure the same properties described in Sections \ref{subsec:hinf_bound}, and hold jointly with confidence $1-2\beta$.

\section{Variance of $\bar{w}(\omega_i)$}
\label{appendix:noise_radius}
By Parseval's theorem, we have
\begin{equation}
\sum_{k=0}^{k=T-1} w_k w_k^\top = \frac{1}{T}\sum_{i=0}^{N-1} \mathbf{w}(e^{j\omega_i}) \mathbf{w}(e^{j\omega_i})^\top
\end{equation}
where $\omega_i=\frac{i}{T},\,i=0,\dots,T-1$. Given that $\bar{w}(\omega_i)=\frac{\mathbf{w}(e^{j\omega_i})}{T}$, we have
\begin{equation*}
\sum_{k=0}^{k=T-1} w_k w_k^\top = T\sum_{i=0}^{N-1} \bar{w}(\omega_i) \bar{w}(\omega_i)^\top.
\end{equation*}
From this, we get $\sigma_\mathrm{w}^2 I=T\sigma_{\bar{\mathrm{w}}}^2 I$. Therefore, the variance of $\bar{w}(\omega_i)$ is given by
\begin{equation}
\sigma_{\bar{\mathrm{w}}}^2 I =\frac{\sigma_\mathrm{w}^2}{T}I.
\end{equation}


\section{Proof of Lemma \ref{lemma:robust}}
\label{appendix:robust_lemma}
\begin{proof}
The proof follows the arguments in \cite[Prop. 3.13, Sec. 4.2.5]{scherer2000linear} and \cite{scherer2001lpv} for LPV control. First, note that the set definitions $\mathbf{\Delta}_\mathrm{s},\mathbf{\Delta}_\mathrm{u}$ can be equivalently represented as $\Delta_\mathrm{s}\in\mathbf{\Delta}_\mathrm{s}:=\{\Delta: \lambda_\mathrm{s} R_\mathrm{s}-\lambda_\mathrm{s}\Delta^\top  \Delta \succ 0\}$ and $\Delta_\mathrm{u}~\in~\mathbf{\Delta}_\mathrm{u}~:=~\{\Delta:~\lambda_\mathrm{u} R_\mathrm{u}-\lambda_\mathrm{u}\Delta^\top \Delta  \succ 0\}$ with arbitrary positive scalars $\lambda_\mathrm{s}, \lambda_\mathrm{u}>0$. Define $X=N^{-1}$ and $K_\mathrm{x}=M N^{-1}$.
Multiplying the Schur complement of inequality \eqref{eq:LMI_gs1} by $\text{diag}(N^{-1},I,I,I)$ and inequality \eqref{eq:LMI_gs2} by $\text{diag}(N^{-1},I)$, respectively, from the left and the right yields

\begin{small}
\begin{subequations}\label{eq:h2_biginequality}
\begin{align}\label{eq:biginequality}
\nonumber
\left[ \begin{array}{c}
*\\ *\\ \hline
*\\ *\\ \hline
*\\ *\\ \hline
*
\end{array}\right]^\top &
\left[
\begin{array}{c|c|c|c}
\begin{matrix} -X&0\\0&X \end{matrix} & 
\begin{matrix} 0&0\\0&0\end{matrix} &
\begin{matrix} 0&0\\0&0 \end{matrix} &
\begin{matrix} 0\\0 \end{matrix} \\
\hline
\begin{matrix} 0&0\\0&0 \end{matrix} &
\begin{matrix} \lambda_\mathrm{s} P_\mathrm{s} \end{matrix} & 
\begin{matrix} 0&0\\0&0 \end{matrix} &
\begin{matrix} 0\\0 \end{matrix} \\
\hline
\begin{matrix} 0&0\\0&0 \end{matrix} &
\begin{matrix} 0&0\\0&0 \end{matrix} &
\begin{matrix} \lambda_\mathrm{u} P_\mathrm{u} \end{matrix} & 
\begin{matrix} 0\\0 \end{matrix} \\
\hline
\begin{matrix} 0&0 \end{matrix} &
\begin{matrix} 0&0 \end{matrix} &
\begin{matrix} 0&0  \end{matrix} &
\begin{matrix} -\gamma_\mathrm{p} I \end{matrix}
\end{array}
\right] \\
\times & \left[ \begin{array}{cccc}
I&0&0&0 \\
\hat{A}_0+\hat{B}_0 K_\mathrm{x}& I+ \hat{B}_0 K_\mathrm{s}&I&I \\
\hline
0&I&0&0\\
\begin{bmatrix} I\\K_\mathrm{x} \end{bmatrix} & \begin{bmatrix} 0\\K_\mathrm{s} \end{bmatrix} & 0 & 0 \\
\hline
0&0&I&0 \\
 \begin{bmatrix} I\\K_\mathrm{x} \end{bmatrix} & \begin{bmatrix} 0\\K_\mathrm{s} \end{bmatrix} & 0 & 0\\
\hline
0&0&0&I
\end{array} \right] \prec 0,
\end{align}
\begin{align}
\begin{bmatrix}
X & C^\top\\ C & Z
\end{bmatrix} \succ 0, 
\end{align}
\begin{align}
    \text{trace}(Z) \leq \gamma_\mathrm{p},
\end{align}
\end{subequations}
\end{small}

where $P_\mathrm{s}=\begin{bmatrix}-I&0\\0&R_\mathrm{s}\end{bmatrix}$ and $P_\mathrm{u}=\begin{bmatrix}-I&0\\0&R_\mathrm{u}
\end{bmatrix}$. Using \cite{scherer2000linear, scherer2001lpv}, $H_2$ performance is guaranteed if there exists a positive definite matrix $X=X^\top \succ 0$ and $Z$ satisfying the matrix inequalities \eqref{eq:h2_biginequality}.  $N \succ 0$ follows from \eqref{eq:h2_biginequality} with $N~=~X^{-1}~\succ~0$.\qedhere
\end{proof}
\section*{Acknowledgment}
The authors would like to thank Xavier Bombois for valuable feedback regarding the interpretation of classical identification bounds.
%
%
\bibliographystyle{ieeetr}
\bibliography{lit}  

\begin{IEEEbiography}[{\includegraphics[width=1in,height=1.25in,clip,keepaspectratio]{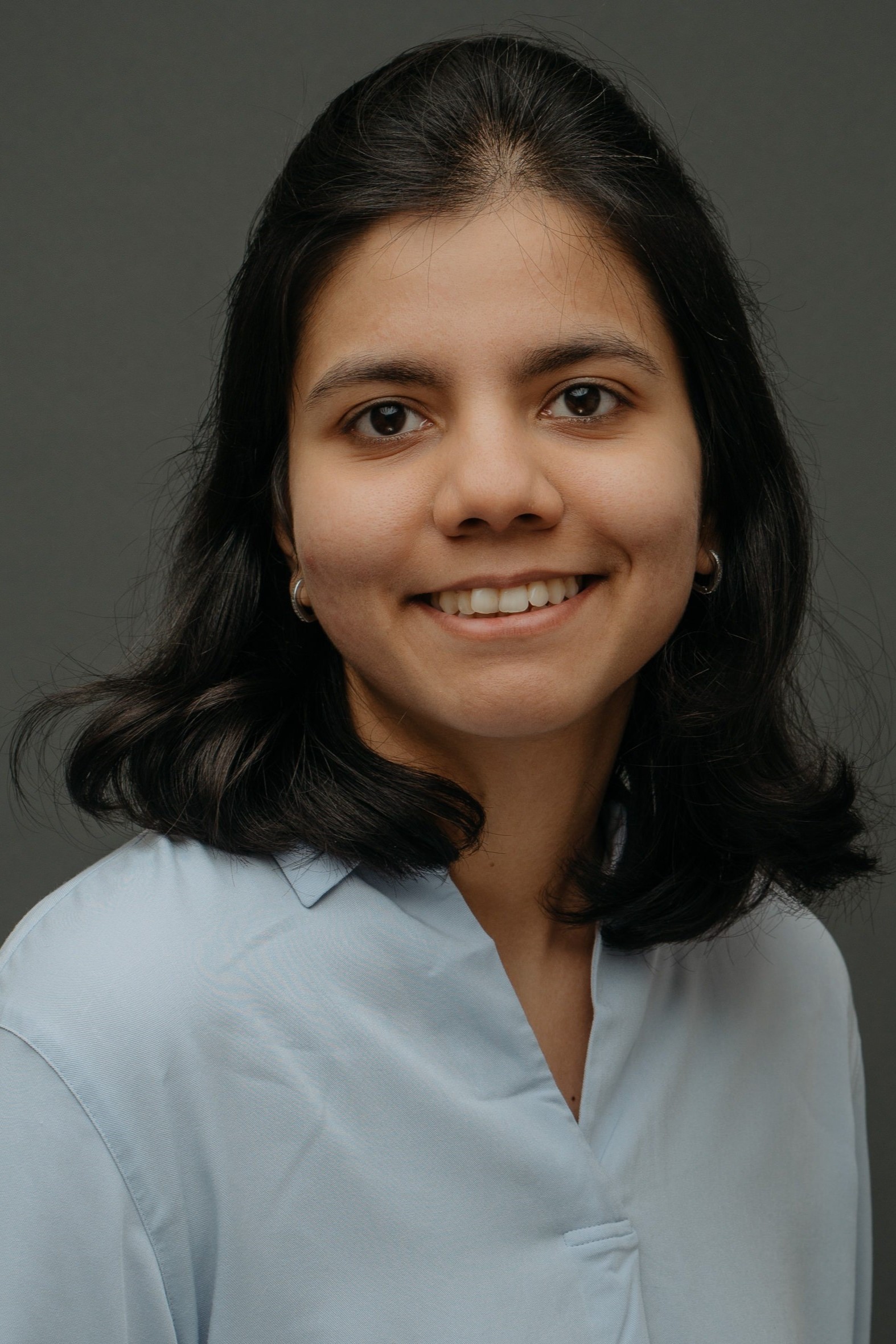}}]{Janani Venkatasubramanian} received her Master degree in Electrical Engineering from the Delft University of Technology, The Netherlands, in 2018. Since May 2019, she is a Ph.D. student at the Institute for Systems Theory and Automatic Control, University of Stuttgart under the supervision of Prof. Frank Allg\"ower, and a member of the International Max Planck Research School for Intelligent Systems (IMPRS-IS). Her research interests lie in the area of learning and adaptive control. \end{IEEEbiography}\vspace{-1 cm}

\begin{IEEEbiography}[{\includegraphics[width=1in,height=1.25in,clip,keepaspectratio]{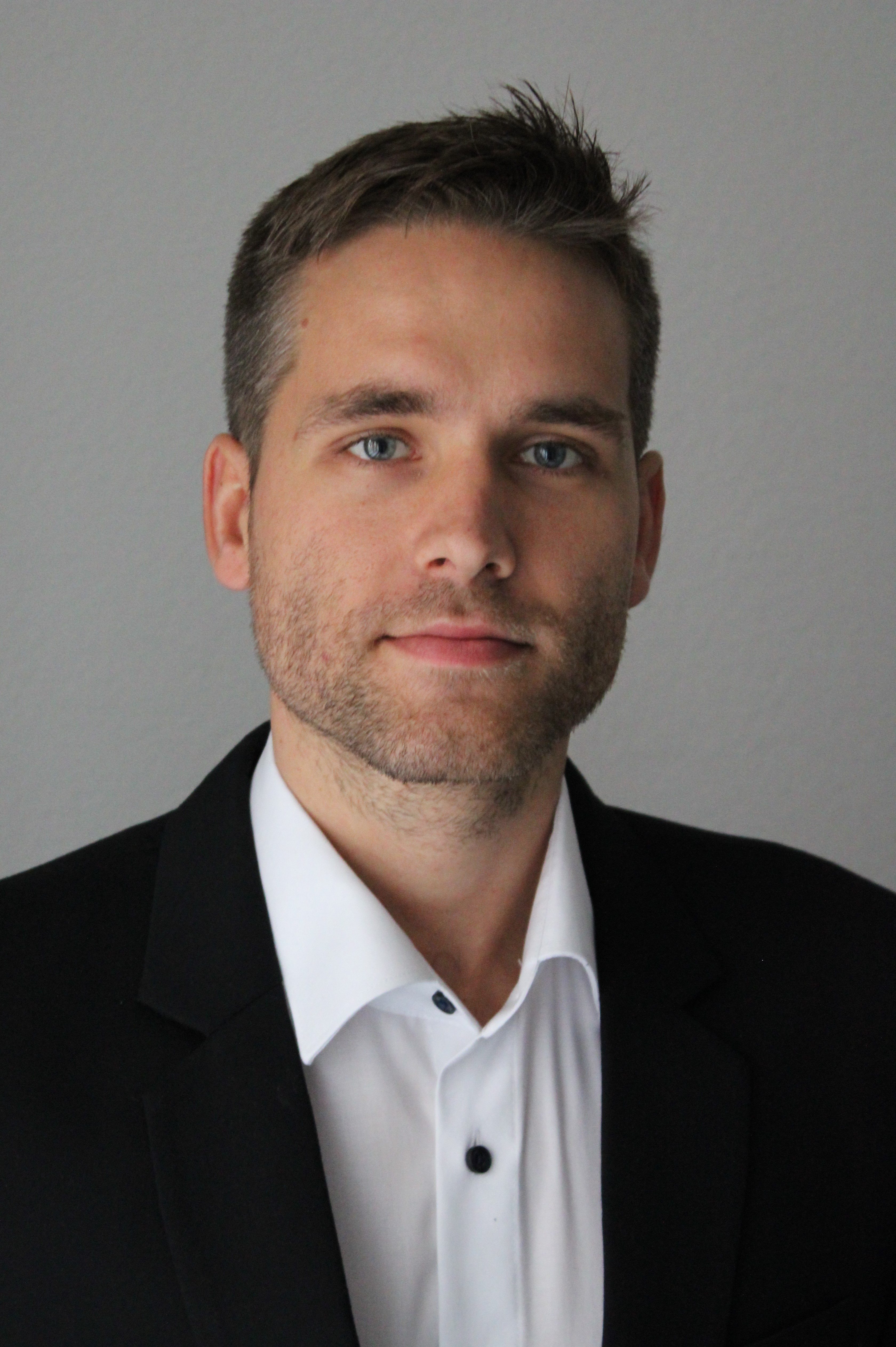}}]{Johannes K\"ohler} received his Master degree in Engineering Cybernetics from the University of Stuttgart, Germany, in 2017. In 2021, he obtained a Ph.D. in mechanical engineering, also from the University of Stuttgart, Germany, for which he received the 2021 European Systems \& Control Ph.D. award. He is currently a postdoctoral researcher at the Institute for Dynamic Systems and Control (IDSC) at ETH Zürich. His research interests are in the area of model predictive control and control and estimation for nonlinear uncertain systems. \end{IEEEbiography}\vspace{-1 cm}

\begin{IEEEbiography}[{\includegraphics[width=1in,height=1.25in,clip,keepaspectratio]{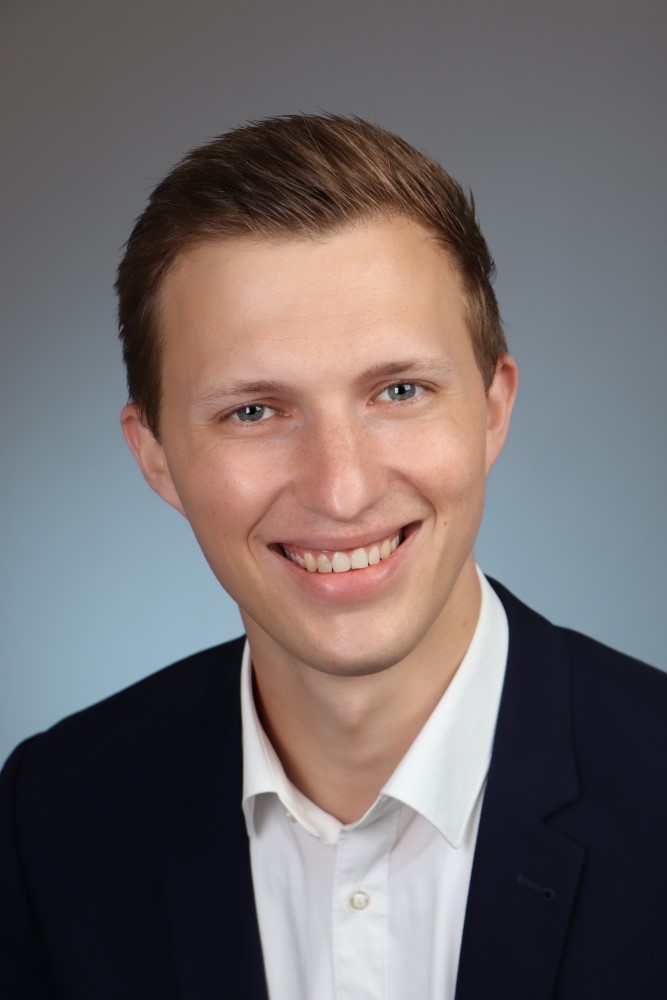}}]{Julian Berberich} received a Master’s degree in Engineering Cybernetics from the University of Stuttgart, Germany, in 2018. In 2022, he obtained a Ph.D. in Mechanical Engineering, also from the University of Stuttgart, Germany. He is currently working as a Lecturer (Akademischer Rat) at the Institute for Systems Theory and Automatic Control at the University of Stuttgart, Germany. In 2022, he was a visiting researcher at the ETH Zürich, Switzerland. He has received the Outstanding Student Paper Award at the 59th IEEE Conference on Decision and Control in 2020 and the 2022 George S. Axelby Outstanding Paper Award. His research interests include data-driven analysis and control as well as quantum computing. \end{IEEEbiography}\vspace{-1 cm} 

\begin{IEEEbiography}[{\includegraphics[width=1in,height=1.25in,clip,keepaspectratio]{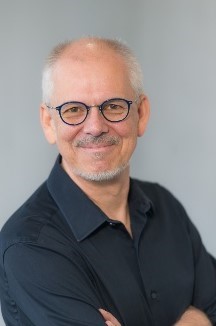}}]{Frank Allg\"ower} is professor of mechanical engineering at the University of Stuttgart, Germany, and Director of the Institute for Systems Theory and
Automatic Control (IST) there. Frank is active in serving the community in several roles: Among others he has been President of the International Federation of Automatic Control (IFAC) for the years 2017-2020, Vice-president for Technical Activities of the IEEE Control Systems Society for 2013/14, and Editor of the journal Automatica from 2001 until 2015. From 2012 until 2020 Frank served in addition as Vice-president for the German Research Foundation (DFG), which is Germany’s most important research funding organization. His research interests include predictive control, data-based control, networked control, cooperative control, and nonlinear control with application to a wide range of fields including systems biology. \end{IEEEbiography}

\end{document}